\documentclass[preprint,12pt]{elsarticle}
\usepackage{tikz}
\usetikzlibrary{shapes,shadows,arrows}
\usepackage[dvipdfm]{hyperref}
\usepackage{verbatim}
\usepackage{graphicx}
\usepackage{amssymb}
\usepackage{algpseudocode}
\usepackage{algorithm}
\usepackage{float}
\usepackage{mathtext}
\usepackage{caption}
\usepackage{footnote}
\usepackage{amsmath}

\usepackage{amsthm}
\usepackage{mathtools}
\usepackage[numbers]{natbib}
\usepackage{mathrsfs}
\usepackage{algorithmicx}
\usepackage{balance}
\usepackage{textcomp}
\usepackage{multicol}
\usepackage{xcolor,colortbl}
\usepackage{pifont}
\usepackage{subcaption}

\def\NoNumber#1{{\def\alglinenumber##1{}\State #1}\addtocounter{ALG@line}{-1}}

\let\oldReturn\Return
\renewcommand{\Return}{\State\oldReturn}
\algnewcommand\algorithmicoutput{\textbf{Output:}}
\algnewcommand\Output{\item[\algorithmicoutput]}
\algnewcommand\algorithmicforeach{\textbf{for each}}
\algdef{S}[FOR]{ForEach}[1]{\algorithmicforeach\ #1\ \algorithmicdo}

\newtheorem{definition}{Definition}
\newtheorem{mylemma}{Lemma}

\newtheorem{proposition}{Proposition}
\newtheorem{corollary}{Corollary}

\makeatletter
\newenvironment{breakablealgorithm}
  {
   \begin{center}
     \refstepcounter{algorithm}
     \hrule height.8pt depth0pt \kern2pt
     \renewcommand{\caption}[2][\relax]{
       {\raggedright\textbf{\ALG@name~\thealgorithm} ##2\par}%
       \ifx\relax##1\relax 
         \addcontentsline{loa}{algorithm}{\protect\numberline{\thealgorithm}##2}%
       \else 
         \addcontentsline{loa}{algorithm}{\protect\numberline{\thealgorithm}##1}%
       \fi
       \kern2pt\hrule\kern2pt
     }
  }{
     \kern2pt\hrule\relax
   \end{center}
  }
\makeatother
\begin{document}

\begin{frontmatter}


 \title{Hiring Expert Consultants in E-Healthcare: A Two Sided Matching Approach}
 \author{Vikash Kumar Singh\fnref{label1}}
 \ead{vikas.1688@gmail.com}
 \author{Sajal Mukhopadhyay\fnref{label1}}
 \ead{sajmure@gmail.com}
   \author{Fatos Xhafa\fnref{label2}}
 \ead{fatos.xhafa@gmail.com}
 \author{Aniruddh Sharma\fnref{label1}}
 \ead{annirudhsharma189@gmail.com}
  \author{Arpan Roy \fnref{label3}\corref{cor2}}
 \ead{arpanroy1994@gmail.com}
 \cortext[cor2]{This research was done while the author was a graduate student at Department of Information Technology, NIT Durgapur, India.}


 \address[label1]{Department of Computer science and Engineering, NIT, Durgapur, India}
  \address[label2]{Department of Computer science, Polytechnic University of Catalonia, Barcelona, Catalonia, Spain}
  \address[label3]{Applications Engineer, Oracle, Bangalore, India}


\begin{abstract}
Very often in some censorious healthcare scenario, there may be a need to have some expert consultancies (especially by doctors) that are not available in-house to the hospitals. Earlier, this interesting healthcare scenario of hiring the ECs (mainly \emph{doctors}) from outside of the hospitals had been studied with the robust concepts of mechanism design with or without money. In this paper, we explore the more \emph{realistic} two sided matching in our set-up, where the members of the two participating communities, namely \emph{patients} and \emph{doctors} are revealing the strict preference ordering over all the members of the opposite community for a stipulated amount of time. We assume that patients and doctors are \emph{strategic} in nature. With the theoretical analysis, we demonstrate that the TOMHECs, that results in \emph{stable} allocation of doctors to the patients is \emph{strategy-proof} (or \emph{truthful}) and \emph{optimal}. The proposed mechanisms are also validated with exhaustive experiments.
\end{abstract}

  
\end{frontmatter}


\section{Introduction}
\label{Introduction}
The \emph{expert advices} or \emph{consultancies} provided by the expert consultants (ECs) mainly doctors can be thought of as one of the most indispensable events that occurs in the hospital(s) or medical unit(s) on a regular basis. Over the past few years, there had been a perplexing growth in the demand of ECs (especially \emph{doctors}) during some critical surgical processes (or operations) that are taking place in the operation theatres (OTs) of the hospitals. The unprecedented growth in the demand of the ECs, has made ECs \emph{busy} and \emph{scarce} in nature. It is to be noted that, this unique nature ($i.e.$ busy and scarce) of ECs in the healthcare lobby provides an edge to the research community in the healthcare domain to think of: \emph{How to manage or schedule these limited (or \emph{scarce}) ECs in the OTs of the hospitals, during some censorious healthcare situation?} In order to answer the above coined question, previously, there had been a spate of research work in the direction of handling the issues of scheduling the in-house ECs especially doctors ~\cite{Carter_2001, Vassilacopoulos_1985} and nurses ~\cite{ Weil_1995} in an efficient and effective manner. In \cite{Carter_2001, Vassilacopoulos_1985, Beaulieu_2000, Wang_2007} different techniques are discussed and presented to schedule the physicians that are in-house to the hospitals in an efficient way for some critical operations that are taking place in the OTs of that hospitals. In past, the work had been also done in the direction of managing the OTs and the hospitals during the patient congestion scenario. The work in ~\cite{cardoen_b_2010, Blake_J_1997, Demeulemeester_E_2010} focuses on the question of: \emph{how to effectively and efficiently plan and schedule the OTs?} In ~\cite{Dexter_2004, Dexter_2003, Demeulemeester_E_2010} the work has been done for allocating OTs on time to increase operating room efficiency.\\
More importantly, in healthcare domain, one scenario that may be thought of as a challenging issue is, say; in certain critical medical cases, there may be a requirement of some external manpower in the form of ECs (mainly doctors) that are not available in-house to the hospitals. Now, the immediate \emph{natural} question that came in the mind is that, \emph{how to have some external expertise mainly in the form of doctors that are not available in-house to the hospitals?} Surprisingly, literature is very limited for this problem in healthcare domain. This interesting situation of taking expert consultancy from outside of the in-house medical unit during
some censorious medical scenario (mainly \emph{surgical} process) was taken care by \emph{Starren et. al.} ~\cite{Starren_JB_2014}. Moreover, the introduction of such a pragmatic field of study in the healthcare domain by \emph{Starren et. al.} ~\cite{Starren_JB_2014} has given rise to several open questions for the researchers, such as: (a) \emph{which ECs are to be considered as the possible expertise provider in the consultancy arena}? (b) \emph{What incentives policies in the form of perks and facilities are to be presented in-front of the ECs, so as to drag as many ECs as possible in the consultancy arena}?\\
In ~\cite{Vikash2015EAS}, the problem of hiring one or more doctors for a patient from outside of the admitted hospital for some critical operation under monetary environment (experts are charging for their services) with the \emph{infinite budget} are addressed. With the consideration that, ECs are having some social connections in real life, \emph{Singh et. al.} ~\cite{DBLP:journals/corr/SinghM16a} considered a budgeted setting of the problem in ~\cite{Vikash2015EAS} motivated by ~\cite{Singer:2014:BFM:2692359.2692366} for hiring \emph{k} doctors out of the available \emph{n} doctors ($k < n$), such that the total payment made to the ECs do not exceed the total budget of a patient.
As opposed to the money involved hiring of ECs as mentioned in ~\cite{Vikash2015EAS, DBLP:journals/corr/SinghM16a}, another market of hiring ECs can be thought of where the ECs are providing their 
expertise free of cost.
Recently, \emph{Singh et. al.} ~\cite{DBLP:journals/corr/SinghM16} have addressed this idea, where the expert services are distributed free of cost. For hiring ECs in \emph{money free} environment ($i.e.$ money is not involved in any sense) they have utilized the idea of one sided strict preference (in this case strict preference from patient side) over the available doctors in the consultancy arena.\\
In this paper, we have tried to model the \emph{ECs hiring problem} as a two sided preference market in healthcare domain motivated by ~\cite{DBLP:journals/corr/abs-1104-2872, Dughmi:2010:TAW:1807342.1807394, Gale_AMM_1962, Shapley_2013}. The idea behind studying the \emph{ECs hiring problem} as a more appealing two sided preference market is that, in this environment, the members present in two different communities have the privilege to provide the strict preference ordering over all the available members of the opposite community. For example, in our case, we have two communities (or parties) in the consultancy arena: (a) \emph{Patient party} (b) \emph{Doctor party}. So, the members of the \emph{patient party} provide strict preference ordering over all the available members (or the subset of available members) in the \emph{doctor party} and \emph{vice versa}.
\subsection{Our Contributions}
\noindent The main contributions of our work are as follows.\\
$\bullet$ We have tried to model the \emph{ECs hiring problem} as a two sided matching problem in healthcare domain. \\
$\bullet$ We propose two mechanisms: a naive approach $i.e.$ \emph{\textbf{ra}ndomized \textbf{m}echanism for \textbf{h}iring \textbf{e}xpert \textbf{c}onsultants} (RAMHECs) and a \emph{truthful} and \emph{optimal} mechanism; namely \emph{\textbf{t}ruthful \textbf{o}ptimal \textbf{m}echanism for \textbf{h}iring \textbf{e}xpert \textbf{c}onsultants} (TOMHECs).\\
$\bullet$ We have also proved that for any instance of \emph{n} patients and \emph{n} doctors the allocation done by TOMHECs results in \emph{stable}, \emph{truthful}, and \emph{optimal} allocation for requesting party.\\
$\bullet$ TOMHECs establish an upper bound of $O(kn^2)$ on the number of iterations required to determine a \emph{stable} allocation for any instance of \emph{n} patients and \emph{n} doctors.\\ 
$\bullet$ A substantial amount of analysis and simulation are done to validate the performance of RAMHECs and TOMHECs via \emph{optimal} allocation measure.

\noindent The remainder of the paper is structured as follows. Section \textbf{2} describes our proposed model. Some required definitions are discussed in section \textbf{3}. Section \textbf{4} illustrates the proposed mechanisms. Further analytic-based analysis of the mechanisms are carried out in section \textbf{5}. A detailed analysis of the experimental results is carried out in section \textbf{6}. Finally, conclusions are drawn and some future directions are depicted in section \textbf{7}.

\section{System model}
 We consider the scenario, where there are multiple hospitals \emph{say} \emph{n} given as $\hbar = \{\hbar_1, \hbar_2, \ldots, \hbar_n\}$. In each hospital $\hbar_i \in \hbar$, there exists several patients with different diseases (in our case patients and doctors are categorized based on the diseases and areas of expertise respectively.) belonging to different income group that requires somewhat partial or complete expert consultancies from outside of the admitted hospitals. By partial expert consultancies it is meant that, the part of expertise from the overall required expert consultancies. The set of \emph{k} different categories is given as: $\boldsymbol{\mathcal{C}} = \{c_1, c_2, \ldots, c_k\}$. The set of all the admitted patients in different categories to different hospitals is given as:
$\boldsymbol{\mathcal{P}} = \bigcup\limits_{\hbar_k \in \hbar} \bigcup\limits_{c_i \in \boldsymbol{\mathcal{C}}} \bigcup\limits_{j=1}^{\mathring{\hbar}_{k}^{i}} p_{i(j)}^{\hbar_k}$
 where $p_{i(j)}^{\hbar_{k}}$ is the patient \emph{j} belonging to $c_i$ category admitted to $\hbar_{k}$ hospital. The expression $\mathring{\hbar}_{k}^{i}$ in term $p_{i(\mathring{\hbar}_{k}^{i})}^{\hbar_{k}}$ indicates the total number of patients in hospital $\hbar_k$ belonging to $c_i$ category. The patients who need
consultancy may belong to different income bars.
So, in this
scenario, each hospital tries to select the patient from the
lowest income bar in a particular category (say $c_i$ category)
who will get the free consultation.
On the other hand, there are several doctors having different expertise associated with different hospitals say $\boldsymbol{\mathcal{H}} = \{\mathcal{H}_1, \mathcal{H}_2, \ldots, \mathcal{H}_n\}$. The set of all the available doctors in different categories associated with different hospitals is given as: $\boldsymbol{\mathcal{D}} = \bigcup\limits_{\mathcal{H}_k \in \mathcal{H}} \bigcup\limits_{c_j \in \boldsymbol{\mathcal{C}}} \bigcup\limits_{i=1}^{\mathring{\mathcal{H}}_{k}^{i}} d_{j(i)}^{\mathcal{H}_k}$ 
 where $d_{i(j)}^{\mathcal{H}_k}$ is the doctor \emph{j} belonging to $c_i$ category associated to $\mathcal{H}_k$ hospital. The expression $\mathring{\mathcal{H}}_{k}^{j}$ in term $d_{j(\mathring{\mathcal{H}}_{k}^{j})}^{\mathcal{H}_{k}}$ indicates the total number of doctors associated with hospital $\mathcal{H}_k$ in $c_j$ category. Our model captures only a single category say $c_i$ but it is to be noted that our proposed model works well for the system considering multiple categories simultaneously. Only thing is that we have to repeat the process \emph{k} times as \emph{k} categories are existing. For simplicity purpose, in any category $c_i$ we have considered the number of patients and number of doctors are same $i.e.$ \emph{n} along with an extra constraint that each of the members of the participating parties provides a strict preference ordering over all the available members of the opposite party. But, one can think of the situation where there are \emph{n} number of patients and \emph{m} number of doctors in a category such that $m \neq n$ ($m>n$ or $m<n$). Moreover, the condition that every members of the participating party is providing the strict preference ordering over all the available members of the opposite community is not essential and can be relaxed for all the three cases ($i.e.$ $m=n$, $m<n$, and $m>n$). 
 \begin{figure}[H]
\centering
 \includegraphics[scale=0.80]{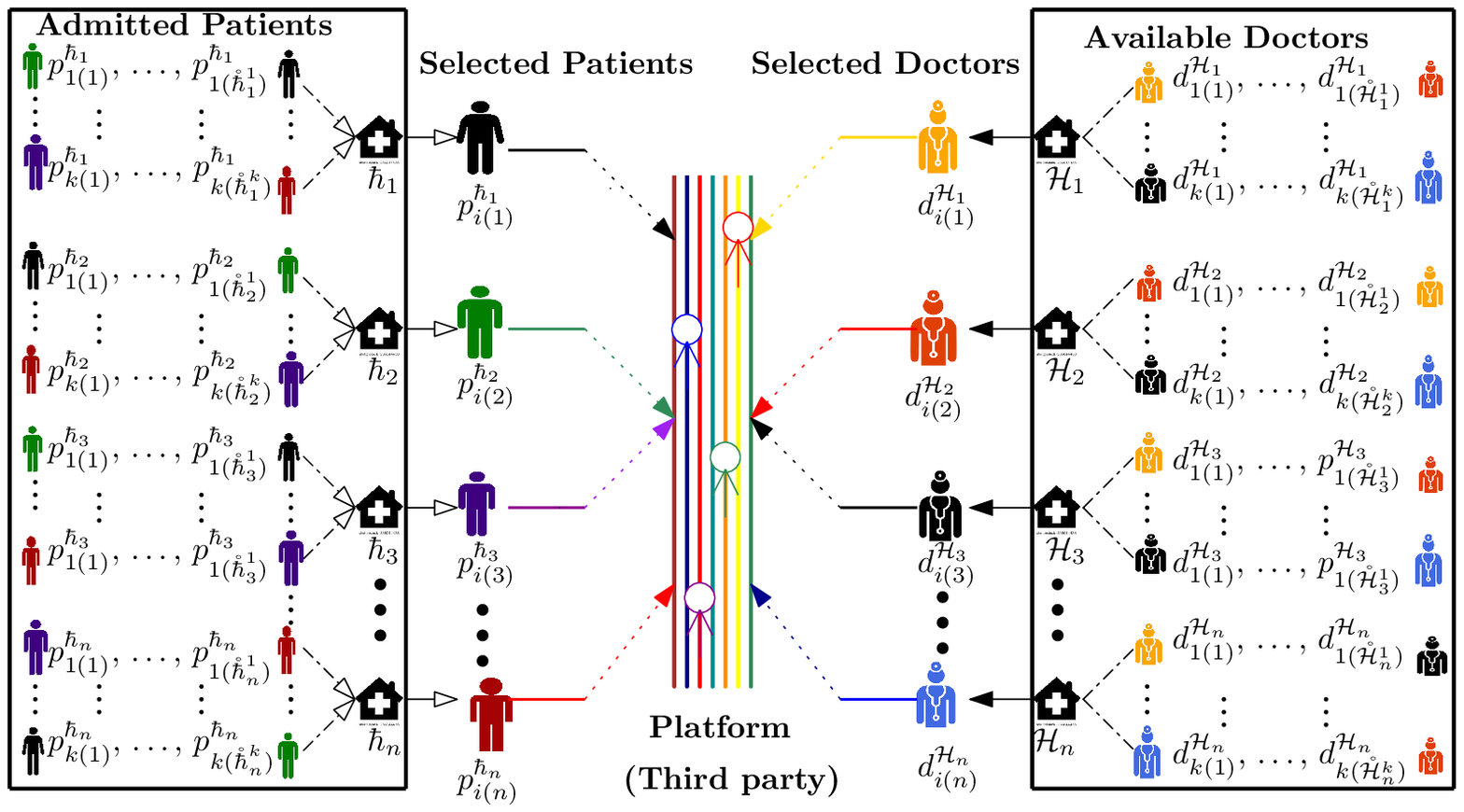}
 \caption{System model}
 \label{fig:1}
\end{figure}  
By relaxation, it is meant that the members of the participating parties may provide the strict preference ordering over the subset of the members of the opposite party.    
At a
particular time several doctors ($>$ n) are providing their willingness to impart free consultancy to some patients present
in the consultancy arena in a particular category as shown
in the right side of the Figure \ref{fig:1}. In the schematic diagram shown in Figure \ref{fig:1}, for representation purpose one doctor is selected from all the interested doctors from each hospital belonging to a particular category $c_i$. 
But in general one can think of the situation where multiple doctors can be selected from the available doctors from a particular hospital in a particular category $c_i$ such that $|\boldsymbol{\mathcal{P}}_i| = \sum_{\mathcal{H}_j \in \boldsymbol{\mathcal{H}}} \dot{\mathcal{H}}_{j}^{i}$; where $0 \leq \dot{\mathcal{H}}_{j}^{i} \leq n $ is the number of doctors selected from hospital $\mathcal{H}_j$ in $c_i$ category and placed into the consultancy arena.
Following the above discussed criteria, the third party selects
$n$ doctors out of all the doctors in a particular category $c_i$ as
a possible expert consultant and is given as $\boldsymbol{\mathcal{D}}_{i} = \{d_{i(1)}^{\mathcal{H}_1}, d_{i(2)}^{\mathcal{H}_2}, \ldots, d_{i(n)}^{\mathcal{H}_n} \}$ and a set of selected patients from $c_i$ category is given as $\boldsymbol{\mathcal{P}}_{i} = \{p_{i(1)}^{\hbar_{1}}, p_{i(2)}^{\hbar_{2}}, \ldots, p_{i(n)}^{\hbar_{n}}\}$. If not specified explicitly, \emph{n} denotes the total number of patients and the total number of doctors that are participating in the consultancy arena in any category $c_i$. For placing \emph{n} doctors in the
consultancy arena from the available doctors, the third party
can take the help of the qualification of the doctors and number of successful operations or consultancies given so far by
that doctor. 
Each patient $p_{i(j)}^{\hbar_k}$ reveals a strict preference ordering over the participating set of doctors $\boldsymbol{\mathcal{D}}_i$ in a category $c_{i}$ and also each doctor $d_{i(t)}^{\mathcal{H}_k}$ provides the strict preference ordering over the set of participating patients of category $c_{i}$ in the consultancy arena. The strict preference ordering of the patient $p_{i(k)}^{\hbar_{j}}$ over the set of doctors $\boldsymbol{\mathcal{D}}_{i}$ is denoted by $\succ_{k}^{i}$. More formally, the significance of $d_{i(\ell)}^{\mathcal{H}_j} \succ_{k}^{i} d_{i(m)}^{\mathcal{H}_k}$ is that the patient $p_{i(k)}^{\hbar_t}$ ranks doctor $d_{i(\ell)}^{\mathcal{H}_j}$ above the doctor $d_{i(m)}^{\mathcal{H}_k}$. The preference profile of all the patients for \emph{k} different categories is denoted as $\succ = \{\succ^{1}, \succ^{2}, \ldots, \succ^{k}\}$, where $\succ^{i}$ denotes the preference profile of all the patients in category $c_i$ over all the doctors in set $\boldsymbol{\mathcal{D}}_{i}$ represented as $\succ^{i} = \{\succ_{1}^{i}, \succ_{2}^{i}, \ldots, \succ_{n}^{i}\}$. The preference profile of all the patients in $c_i$ category except the patient \emph{r} is given as $\succ_{-r}^{i} = \{\succ_{1}^{i}, \succ_{2}^{i}, \ldots, \succ_{r-1}^{i}, \succ_{r+1}^{i}, \ldots \succ_{n}^{i}\} $. On the other hand, the doctors may give preferences based on the location
where he/she (henceforth he) and the patients are located. The strict preference ordering of the doctor $d_{j(t)}^{\mathcal{H}_k}$ is denoted by $\curlyeqsucc_{j}^{t}$ over the set of patients $\boldsymbol{\mathcal{P}}_{j}$, where $p_{j(\ell)}^{\hbar_k} \curlyeqsucc_{j}^{t} p_{j(m)}^{\hbar_i}$ means that doctor $d_{j(t)}^{\mathcal{H}_k}$ ranks $p_{j(\ell)}^{\hbar_k}$ above $p_{j(m)}^{\hbar_i}$. The set of preferences of all the doctors in \emph{k} different categories is denoted as $\curlyeqsucc = \{\curlyeqsucc_{1}, \curlyeqsucc_{2}, \ldots, \curlyeqsucc_{k}\}$, where $\curlyeqsucc_{j}$ contains the strict preference ordering of all the doctors in $c_j$ category over all the patients in set $\boldsymbol{\mathcal{P}}_{j}$ represented as $\curlyeqsucc_{j} = \{\curlyeqsucc_{j}^{1}, \curlyeqsucc_{j}^{2}, \ldots, \curlyeqsucc_{j}^{n}\}$. The strict preference ordering of all the doctors in $c_j$ category except the doctor \emph{s} is represented as $\curlyeqsucc_{j}^{-s} = \{\curlyeqsucc_{j}^{1}, \curlyeqsucc_{j}^{2}, \ldots, \curlyeqsucc_{j}^{s-1}, \curlyeqsucc_{j}^{s+1}, \ldots, \curlyeqsucc_{j}^{n}\}$.
It is to be noted that the allocation of the doctors to the patients for category $c_i$ under consideration is captured by the allocation function $\mathcal{A}_i$: $\succ$ $\times$ $\curlyeqsucc$ $\rightarrow$ $\boldsymbol{\mathcal{P}}_i$ $\times$ $\boldsymbol{\mathcal{D}}_i$. The resulting allocation vector is given as  $\boldsymbol{\mathcal{A}} = \{\mathcal{A}_{1}, \mathcal{A}_{2}, \ldots, \mathcal{A}_{k}\}$. Each allocation vector
$\mathcal{A}_{i} \in \boldsymbol{\mathcal{A}}$ denotes the patient-doctor pairs belonging to
the $c_i$ category denoted as $\mathcal{A}_{i} = \bigcup_{j=1;k=1}^{n;n} a_{jk}^{i}$, where each $a^{i}_{lm} \in \mathcal{A}_{i}$ is a pair $\{p_{i(l)}^{\hbar_k}, d_{i(m)}^{\mathcal{H}_j}\}$. The matching between the patients and doctors for any category $c_i$ is captured by the  
mapping function $\boldsymbol{\mathcal{M}}: \boldsymbol{\mathcal{P}}_i \cup \boldsymbol{\mathcal{D}}_i$ \textrightarrow $\boldsymbol{\mathcal{D}}_i \cup \boldsymbol{\mathcal{P}}_i$. More formally, $
\boldsymbol{\mathcal{M}}(p_{i(j)}^{\hbar_k}) = d_{i(k)}^{\mathcal{H}_\ell} $ means that patient $p_{i(j)}^{\hbar_k}$ is matched to $d_{i(k)}^{\mathcal{H}_\ell}$ doctor and 
$\boldsymbol{\mathcal{M}}(d_{i(\ell)}^{\mathcal{H}_k}) = p_{i(t)}^{\hbar_\ell} $ means that doctor $d_{i(\ell)}^{\mathcal{H}_k}$ is matched to $p_{i(t)}^{\hbar_\ell}$.  
\section{Required definitions}

\begin{definition}\label{d2}
\textbf{Blocking pair}: Fix a category $c_k$. We say that a pair $p_{k(i)}^{\hbar_t}$ and $d_{k(j)}^{\mathcal{H}_{\ell}}$ form a \emph{blocking pair} for matching $\boldsymbol{\mathcal{M}}$, if the following three conditions holds: (i) $\boldsymbol{\mathcal{M}}(p_{k(i)}^{\hbar_t}) \neq d_{k(j)}^{\mathcal{H}_{\ell}}$, (ii) $d_{k(j)}^{\mathcal{H}_{\ell}}$ $\succ_{i}^{k}$ $\boldsymbol{\mathcal{M}}(p_{k(i)}^{\hbar_t})$, and (iii) $p_{k(i)}^{\hbar_t}$ $\curlyeqsucc_{k}^{j}$ $\boldsymbol{\mathcal{M}}(d_{k(j)}^{\mathcal{H}_{\ell}})$.
\end{definition}
\begin{definition}\label{d3}
\textbf{Stable matching}: Fix a category $c_k$. A matching $\boldsymbol{\mathcal{M}}$ is \emph{stable} if there is no pair $p_{k(i)}^{\hbar_t}$ and $d_{k(j)}^{\mathcal{H}_{\ell}}$ such that it satisfies the conditions mentioned in (i)-(iii) in Definition 1. 
\end{definition}
\begin{definition}\label{d4}
\textbf{Perfect matching}: Fix a category $c_k$. A matching $\boldsymbol{\mathcal{M}}$ is \emph{perfect matching} if there exists one-to-one matching between the members of $\boldsymbol{\mathcal{P}}_k$ and $\boldsymbol{\mathcal{D}}_k$.  
\end{definition}
\begin{definition} \label{d5}
\textbf{Patient-optimal stable allocation}: Fix a category $c_k$. A matching $\boldsymbol{\mathcal{M}}$ is \emph{patient optimal}, if there exists no stable matching $\boldsymbol{\mathcal{M}}'$ such that $\boldsymbol{\mathcal{M}}'(p_{k(j)}^{\hbar_t}) \succ_{j}^{k} \boldsymbol{\mathcal{M}}(p_{k(j)}^{\hbar_t})$ or $\boldsymbol{\mathcal{M}}'(p_{k(j)}^{\hbar_t}) =_{j}^{k} \boldsymbol{\mathcal{M}}(p_{k(j)}^{\hbar_t})$ for at least one $p_{i(j)}^{\hbar_t} \in \boldsymbol{\mathcal{P}}_i$. Similar is the situation for \emph{doctor-optimal stable allocation}.
\end{definition}
\begin{definition}\label{d1}
\textbf{Strategy-proof for requesting party}: Fix a category $c_k$. Given the preference profile $\succ^{k}$ and $\curlyeqsucc_k$ of the patients and doctors in $c_k$ category, a mechanism $\mathbb{M}$ is \emph{strategy-proof} (\emph{truthful}) for the requesting party if for each members of the requesting party $\mathcal{A}_k$ is preferred over $\hat{\mathcal{A}}_k$; where $\hat{\mathcal{A}}_k$ is the allocation when at least one member in requesting party is misreporting.      
\end{definition}

\section{Proposed mechanisms}
The idea behind proposing randomized mechanism  $i.e.$ RAMHECs is to better understand the more robust and philosophically strong \emph{optimal} mechanism TOMHECs. The further illustration of the mechanisms are done under the consideration that patient party is requesting. Moreover, one can utilize the same road map of the mechanisms by considering doctors as the requesting party. This can easily be done by just interchanging their respective roles in the mechanisms.   
\subsection{RAMHECs}
The idea lies behind the construction of \emph{initialization phase} is to handle the system consisting \emph{k} different categories. The algorithm is depicted in Algorithm 1.
\subsubsection{Upper bound analysis of RAMHECs}
The overall running time of RAMHECs is $O(1) + O(kn) = O(kn)$.
\subsubsection{Correctness of RAMHECs}
The correctness of RAMHECs is proved with the loop invariant technique ~\cite{Coreman_2009, Gries_2009}. The \emph{loop invariant} that we have to prove is that at the end of the $i^{th}$ iteration each of the patients in $c_j$ category gets one distinct doctor allocated. We must show three things for the \emph{loop invariant} technique to be true.\\
\textbf{Initialization:} It is true prior to the first iteration of the \emph{while} loop. Just before the first iteration of the \emph{while} loop $\mathcal{A}_j \leftarrow \phi$. This confirms that $\mathcal{A}_j$ contains no patient-doctor pair prior to the first iteration of the \emph{while} loop.\\
\textbf{Maintenance:} The \emph{loop invariant} to be true, we have to show that if it is true before each iteration of \emph{while} loop, it remains true before the next iteration. The body of the \emph{while} loop allocates a doctor to a patient in a particular category $i.e.$ each time $\mathcal{A}_j$ is incremented by 1. Just before the $i^{th}$ iteration, the $\mathcal{A}_j$ data structure contains $(i-1)$ number of patient-doctor pairs. After the $i^{th}$ iteration, the $\mathcal{A}_j$ data structure contains $i$ patient-doctor pairs. This way at the end of the $i^{th}$ iteration all the $i$ patients gets a distinct doctor and the \emph{patient-doctor} pairs are stored in $\mathcal{A}_j$ $[1 \ldotp \ldotp i]$.\\
\textbf{Termination:} The third property is to check, what happens when the \emph{while} loop terminates. The condition causing the \emph{while} loop to terminate is that, for any category $c_j$, each of the patients are allocated with one distinct doctor leading to \emph{n} patient-doctor pairs in $\mathcal{A}_j$ data structure. Because each loop iteration increments $\mathcal{A}_j$ by 1, we must have $|\mathcal{A}_j| = n$ when all \emph{n} patients are already processed. So, when the loop terminates we have a data structure $\mathcal{A}_j$ $[1 \ldotp \ldotp n]$ that is already processed and it consists of \emph{n} patient-doctor pairs.\\
If the RAMHECs is true for a particular category $c_j \in \boldsymbol{\mathcal{C}}$ it will remain true when all category in $\boldsymbol{\mathcal{C}}$ taken simultaneously. Hence, the RAMHECs is correct.
\begin{algorithm}[H]
\caption{RAMHECs ($\boldsymbol{\mathcal{D}}$, $\boldsymbol{\mathcal{P}}$, $\boldsymbol{\mathcal{C}}$, $\succ$, $\curlyeqsucc$)}
\begin{algorithmic}[1]
      \Output $\boldsymbol{\mathcal{A}} = \{\mathcal{A}_1, \mathcal{A}_2, \ldots, \mathcal{A}_k\}$
	\State \textbf{begin}
	\NoNumber{/* \textbf{Initialization phase} */ }
	\State $\boldsymbol{\mathcal{A}} \leftarrow \phi$
	\For{each $c_i \in  \boldsymbol{\mathcal{C}}$}
	\State $k$ $\leftarrow$ $0$, $i$ $\leftarrow$ $0$, $d^{*}$ $\leftarrow$ $\phi$, $p^{*}$ $\leftarrow$ $\phi$, $\mathcal{A}_i \leftarrow \phi$, $\mathcal{P}^* \leftarrow \phi$, $\mathcal{D}^* \leftarrow \phi$
	\State $i$ $\leftarrow$ \emph{select}($\boldsymbol{\mathcal{P}}$) \Comment{return the index of patient from patient set.}
   \State $\mathcal{P}^*$ $\leftarrow$ $\mathcal{P}_i$
   \State $i$ $\leftarrow$ \emph{select}($\boldsymbol{\mathcal{D}}$) \Comment{return the index of doctor from doctor set.}
   \State $\mathcal{D}^{*}$ $\leftarrow$ $\mathcal{D}_i$
   \NoNumber{/* \textbf{Allocation phase} */}
	      \While {$|\mathcal{A}_i| \neq n$}
	      \State $t \leftarrow rand(\mathcal{P}^*)$ \Comment{return index of randomly selected patient.}
	      \State $p^{*}$ $\leftarrow$ $p_{i(t)}^{\hbar_k}$
	      \State $k$ $\leftarrow$ $rand$($\succ_{t}^i$, $\mathcal{D}^*$)  \Comment{returns the index of randomly selected doctor from the patient \emph{t} preference list.}
	      \State $d^{*}$ $\leftarrow$ $d_{i(k)}^{\mathcal{H}_{\ell}}$
	      \State $\mathcal{A}_{i}$ $\leftarrow$ $\mathcal{A}_i$ $\cup$ $\{(p^{*}$, $d^{*})\}$
	      \State $\mathcal{P}_i$ $\leftarrow$ $\mathcal{P}_i$ $\setminus$ $p^{*}$ \Comment{Removes the allocated patients from available patient list.}
	      \State $\mathcal{D}_i$ $\leftarrow$ $\mathcal{D}_i$ $\setminus$ $d^{*}$ \Comment{Removes the allocated doctors from available doctor list.}
	      \EndWhile
	      \State $\boldsymbol{\mathcal{A}}$ $\leftarrow$ $\boldsymbol{\mathcal{A}}$ $\cup$ $\mathcal{A}_i$
	 \EndFor      
	      \Return $\boldsymbol{\mathcal{A}}$
  \State \textbf{end}
\end{algorithmic}
\end{algorithm}

\subsubsection{Illustrative example}
For understanding purpose, let the category be $c_3$ (say \emph{eye surgery}). The set of patient from 4 different hospitals $\hbar$ = $\{\hbar_1, \hbar_2, \hbar_3, \hbar_4\}$ is given as: $\boldsymbol{\mathcal{P}}_{3} = \{p_{3(1)}^{\hbar_2}, p_{3(2)}^{\hbar_3}, p_{3(3)}^{\hbar_4}, p_{3(4)}^{\hbar_1}\}$. The set of available doctors engaged to 4 different hospitals $\mathcal{H}$ = $\{\mathcal{H}_1, \mathcal{H}_2, \mathcal{H}_3, \mathcal{H}_4\}$ is given as: $\boldsymbol{\mathcal{D}}_3$ = $\{d_{3(1)}^{\mathcal{H}_3}, d_{3(2)}^{\mathcal{H}_1}, d_{3(3)}^{\mathcal{H}_4}, d_{3(4)}^{\mathcal{H}_2}\}$.
The preference profile of patient set $\mathcal{P}_{3}$ is given as: $p_{3(1)}^{\hbar_2}$ = [$d_{3(4)}^{\mathcal{H}_2} \succ_{1}^{3} d_{3(3)}^{\mathcal{H}_4} \succ_{1}^{3} d_{3(1)}^{\mathcal{H}_3} \succ_{1}^{3} d_{3(2)}^{\mathcal{H}_1}$], $ p_{3(2)}^{\hbar_3}$ = [$d_{3(3)}^{\mathcal{H}_4} \succ_{2}^{3} d_{3(4)}^{\mathcal{H}_2} \succ_{2}^{3} d_{3(2)}^{\mathcal{H}_1} \succ_{2}^{3} d_{3(1)}^{\mathcal{H}_3}$], $p_{3(3)}^{\hbar_4}$ = [$d_{3(4)}^{\mathcal{H}_2} \succ_{3}^{3} d_{3(2)}^{\mathcal{H}_1} \succ_{3}^{3} d_{3(1)}^{\mathcal{H}_3} \succ_{3}^{3} d_{3(3)}^{\mathcal{H}_4}$], $ p_{3(4)}^{\hbar_1}$ = [$d_{3(2)}^{\mathcal{H}_1} \succ_{4}^{3} d_{3(3)}^{\mathcal{H}_4} \succ_{4}^{3} d_{3(4)}^{\mathcal{H}_2} \succ_{4}^{3} d_{3(1)}^{\mathcal{H}_3}$]. Similarly, the preference profile of doctor set $\mathcal{D}_3$ is given as: $d_{3(1)}^{\mathcal{H}_3}$ = [$p_{3(1)}^{\hbar_2} \curlyeqsucc_{3}^{1} p_{3(2)}^{\hbar_3} \curlyeqsucc_{3}^{1} p_{3(4)}^{\hbar_1} \curlyeqsucc_{3}^{1} p_{3(3)}^{\hbar_4}$], $d_{3(2)}^{\mathcal{H}_1}$ = [$p_{3(2)}^{\hbar_3} \curlyeqsucc_{3}^{2} p_{3(4)}^{\hbar_1} \curlyeqsucc_{3}^{2} p_{3(1)}^{\hbar_2} \curlyeqsucc_{3}^{2} p_{3(3)}^{\hbar_4}$], $d_{3(3)}^{\mathcal{H}_4}$ = [$p_{3(3)}^{\hbar_4} \curlyeqsucc_{3}^{3} p_{3(1)}^{\hbar_2} \curlyeqsucc_{3}^{3} p_{3(2)}^{\hbar_3} \curlyeqsucc_{3}^{3} p_{3(4)}^{\hbar_1}$], $d_{3(4)}^{\mathcal{H}_2}$ = [$p_{3(4)}^{\hbar_1} \curlyeqsucc_{3}^{4} p_{3(3)}^{\hbar_4} \curlyeqsucc_{3}^{4} p_{3(1)}^{\hbar_2} \curlyeqsucc_{3}^{4} p_{3(2)}^{\hbar_3}$]. Given the preference profiles, 
\emph{while} loop in line 9-17 randomly selects patient $p_{3(3)}^{\hbar_4}$ from the available patients list $\mathcal{P}_3$. Line 12 of the RAMHECs randomly selects doctor $d_{3(4)}^{\mathcal{H}_2}$ from the available preference ordering of $p_{3(3)}^{\hbar_4}$. At the end of first iteration of the \emph{while} loop, the RAMHECs  captures $(p_{3(3)}^{\hbar_4}, d_{3(4)}^{\mathcal{H}_2})$ pair in the $\mathcal{A}_3$ data structure.        
 In the similar fashion, the remaining allocation is done. The final patient-doctor allocation pair done by the mechanism is $\mathcal{A}_3$ = $\{(p_{3(3)}^{\hbar_4}, d_{3(4)}^{\mathcal{H}_2}), (p_{3(2)}^{\hbar_3}, d_{3(3)}^{\mathcal{H}_4}), (p_{3(4)}^{\hbar_1}, d_{3(1)}^{\mathcal{H}_3}), (p_{3(1)}^{\hbar_2}, d_{3(2)}^{\mathcal{H}_1})\}$.

\subsection{TOMHECs}
Our main focus is to propose a mechanism that satisfy the two important economic properties: \emph{truthfulness}, and \emph{optimality}	.
 The TOMHECs is illustrated in Algorithm 2.
\subsubsection{Running time}
The total running time of TOMHECs is given as: $T(n)= \sum_{i=1}^{k}(O(1) + (\sum_{i=1}^{n} O(n)))= O(kn^2)$
\\
 \begin{breakablealgorithm}
 \caption{TOMHECs ($\boldsymbol{\mathcal{D}}$, $\boldsymbol{\mathcal{P}}$, $\boldsymbol{\mathcal{C}}$, $\succ$, $\curlyeqsucc$)}
 \begin{algorithmic}[1]
        \Output $\boldsymbol{\mathcal{A}} = \{\mathcal{A}_1, \mathcal{A}_2, \ldots, \mathcal{A}_k\}$
 	\State \textbf{begin}
        \NoNumber{/* \textbf{Initialization phase} */ }
        \State $i$ $\leftarrow$ $0$, $\boldsymbol{\mathcal{A}} \leftarrow \phi$
        \NoNumber{/* \textbf{Allocation phase} */ } 
	\For{each $c_i \in  \boldsymbol{\mathcal{C}}$}
	\State $\mathcal{A}_i \leftarrow \phi$
	\State $i$ $\leftarrow$ \emph{select}($\boldsymbol{\mathcal{P}}$)
   \State $\mathcal{P}^*$ $\leftarrow$ $\boldsymbol{\mathcal{P}}_i$
   \State $i$ $\leftarrow$ \emph{select}($\boldsymbol{\mathcal{D}}$)
   \State $\mathcal{D}^{*}$ $\leftarrow$ $\boldsymbol{\mathcal{D}}_i$
 		 \For{each $d_{i(j)}^{\mathcal{H}_k} \in \mathcal{D}^* $}
 	  \State $\Pi(d_{i(j)}^{\mathcal{H}_k}) \leftarrow \phi$  \emph{\Comment{$\Pi(d_{i(j)}^{\mathcal{H}_k})$ data structure keeps track of set of $p_{i(j)}^{\hbar_k} \in \boldsymbol{\mathcal{P}}_i$ requesting to $d_{i(j)}^{\mathcal{H}_k}$.}}
 	 \EndFor
 	 	\While {$|\mathcal{A}_i| \neq n$}
 	   \For{each free patient $p_{i(j)}^{\hbar_k} \in \boldsymbol{\mathcal{P}}_i$}
 	      \State $d^*$ $\leftarrow$ select most preferred doctor from $\succ_{j}^{i}$ not approached till now.
 	      \State $\Pi(d^*) \leftarrow \Pi(d^*)$ $\cup$ $p_{i(j)}^{\hbar_k}$ 
 	   \EndFor
 	   \For{each engaged doctor $d_{i(j)}^{\mathcal{H}_k} \in \boldsymbol{\mathcal{D}}_i$}
 	       \If {$|\Pi(d_{i(j)}^{\mathcal{H}_k})| > 1$}               
 	          \State $p^{*}$ $\leftarrow$ select\_best($\curlyeqsucc_{i}^{j}$, $\Pi(d_{i(j)}^{\mathcal{H}_k})$)
 	          \If {$(p_{i(j)}^{\hbar_k}, d_{i(j)}^{\mathcal{H}_k}) \in \mathcal{A}_i$ \textbf{and} $p^{*} \curlyeqsucc_{i}^{j} p_{i(j)}^{\hbar_k}$}
 	              \State $\mathcal{A}_i$ $\leftarrow$ $\mathcal{A}_i$ $\setminus$ $(p_{i(j)}^{\hbar_k}, d_{i(j)}^{\mathcal{H}_k})$
 	              \State $\mathcal{A}_i \leftarrow \mathcal{A}_i \cup (p^*, d_{i(j)}^{\mathcal{H}_k}) $
 	              \State $\Pi(d_{i(j)}^{\mathcal{H}_k})$ $\leftarrow$ $\Pi(d_{i(j)}^{\mathcal{H}_k}) \setminus \Pi(d_{i(j)}^{\mathcal{H}_k}) - \{p^*\}$
                  \ElsIf{$(p_{i(j)}^{\hbar_k}, d_{i(j)}^{\mathcal{H}_k}) \notin \mathcal{A}_i$} 
                     \State $\mathcal{A}_i \leftarrow \mathcal{A}_i \cup (p^*, d_{i(j)}^{\mathcal{H}_k}) $
 	             \State $\Pi(d_{i(j)}^{\mathcal{H}_k})$ $\leftarrow$ $\Pi(d_{i(j)}^{\mathcal{H}_k}) \setminus \Pi(d_{i(j)}^{\mathcal{H}_k}) - \{p^*\}$
 	          \EndIf
 	       \ElsIf{$|\Pi(d_{i(j)}^{\mathcal{H}_k})| == 1$}
 	          \If{$(\Pi(d_{i(j)}^{\mathcal{H}_k}), d_{i(j)}^{\mathcal{H}_k}) \notin \mathcal{A}_i$}
 	          \State $\mathcal{A}_i \leftarrow \mathcal{A}_i \cup (\Pi(d_{i(j)}^{\mathcal{H}_k}), d_{i(j)}^{\mathcal{H}_k})$   
 	          \EndIf
               \EndIf  
           \EndFor  
        \EndWhile
                 \State $\boldsymbol{\mathcal{A}}$ $\leftarrow$ $\boldsymbol{\mathcal{A}}$ $\cup$ $\mathcal{A}_i$
\EndFor 
	      \Return $\boldsymbol{\mathcal{A}}$    
  \State \textbf{end}
 \end{algorithmic}
 \end{breakablealgorithm}

\subsubsection{Correctness of the TOMHECs}
The correctness of the TOMHECs is proved with the \emph{loop invariant} technique ~\cite{Coreman_2009, Gries_2009}.\\
The \emph{loop invariant}: Fix a category $c_i$. At the start of $\ell^{th}$ iteration of the \emph{while} loop, the number of temporarily processed patient-doctor pairs or in other words the number of patient-doctor pairs held by $\mathcal{A}_i$ is given as: $|\cup_{j=1}^{\ell-1} \mathcal{A}'_j|$, where $\mathcal{A}'_j$ is the net patient-doctor pairs temporarily maintained in the set $\mathcal{A}'_j$ at the $j^{th}$ iteration. So, on an average the number of patients or doctors (whomsoever is greater) that are to be explored in further iterations are $n - |\cup_{j=1}^{\ell-1} \mathcal{A}'_j|$. From the construction of the TOMECs it is clear that after any $\ell^{th}$ iteration this condition holds: $0 \leq n - |\cup_{i=1}^{\ell} \mathcal{A}'_j| \leq n$; where $1 \leq \ell \leq n^2$. The net minimum number of patient-doctor pairs that can be processed temporarily at any iteration may be \emph{zero}. Hence, inequality $0 \leq n - |\cup_{i=1}^{\ell} \mathcal{A}'_j| \leq n$ is always \emph{true}. We must show three things for this \emph{loop invariant} to be true.\\
\textbf{Initialization:} It is true prior to the first iteration of the \emph{while} loop. Just before the first iteration of the \emph{while} loop, in TOMHECs the inequality $0 \leq n - |\cup_{i=1}^{\ell} \mathcal{A}'_j| \leq n$ blows down to $0 \leq n - 0 \leq n$ $\Rightarrow$ $0 \leq n \leq n$ $i.e.$ no patient-doctor pair is temporarily added to $\mathcal{A}_i$ prior to the first iteration of  \emph{while} loop. This confirms that $\mathcal{A}_i$ contains no patient-doctor pair.\\
\textbf{Maintenance:} For the \emph{loop invariant} to be true, if it is true before each iteration of the \emph{while} loop, it remains true before the next iteration of the \emph{while} loop. The body of \emph{while} loop allocates doctor(s) to the patient(s) with each doctor is allocated to a patient; $i.e.$ each time the cardinality of $\mathcal{A}_i$ is either incremented by some amount or remains similar to previous iteration. Just before the $\ell^{th}$ iteration the patient-doctor pairs temporarily added to $A_i$ are $\cup_{i=1}^{\ell-1} \mathcal{A}'_j$. So, one can conclude from here that the number of patient-doctor pairs that are left is given by inequality: $0 \leq n - |\cup_{i=1}^{\ell-1} \mathcal{A}'_j| \leq n$. After the $(\ell-1)^{th}$ iteration, the available number of patient-doctor pair $n - |\cup_{i=1}^{\ell-1} \mathcal{A}'_j| \geq 0$ can be captured under two cases:\\
\textbf{Case 1:} If $|\mathcal{A}_i| = n$:
This case will lead to exhaust all the remaining patient-doctor pair in the current $\ell^{th}$ iteration and no patient-doctor pair is left for the next iteration. The inequality $ n - (|\cup_{i=1}^{\ell-1} \mathcal{A}'_i \cup \mathcal{A}'_{\ell}|)$ = $ n - (|\cup_{i=1}^{\ell} \mathcal{A}'_i |) = n - |\mathcal{A}_i|$ = 0. Hence, it means that all the remaining patient-doctor is absorbed in this iteration and no patient-doctor pair is left for processing.\\
\textbf{Case 2:}
If $|\mathcal{A}_i| < n$: 
This case captures the possibility that there may be the scenario when few patient-doctor pairs from the remaining patient-doctor pairs may still left out; leaving behind some of the pairs for further iterations. So, the inequality $n - (|\cup_{i=1}^{\ell-1} \mathcal{A}'_i \cup \mathcal{A}'_{\ell}|) > 0$ $\Rightarrow$  $n > n - (|\cup_{i=1}^{\ell} \mathcal{A}'_i |) > 0$ is satisfied.\\
\indent From Case 1 and Case 2, at the end of $\ell^{th}$ iteration the loop invariant is satisfied.\\
 \textbf{Termination:} It is clear that in each iteration the cardinality of output data structure $i.e.$ $\mathcal{A}_i$  either incremented by some amount or remains as the 
previous iteration. This indicates that at some $\ell^{th}$ iteration the loop terminates by dissatisfying the condition of the while loop $|\mathcal{A}_i| \neq n$ at line 12. When the loop terminates it is for sure that $|\mathcal{A}_i| = n$. We can say $n - |\cup_{i=1}^{\ell} \mathcal{A}'_i| = 0 \Rightarrow 0 \leq n$. Thus, this inequality indicates that all the \emph{n} patient and doctors in $c_i$ category are processed and each patient allocated a best possible doctor when the loop terminates.\\
If the TOMHECs is true for the $c_i \in \boldsymbol{\mathcal{C}}$ category it will remain true when all category in $\boldsymbol{\mathcal{C}}$ taken simultaneously. Hence, the TOMHECs is correct.                      
\subsubsection{Example}
Considering the initial set-up discussed in section 4.1.3. According to line 13-16 of TOMHECs each of the patients $p_{3(1)}^{\hbar_2}$, $p_{3(2)}^{\hbar_3}$, $p_{3(3)}^{\hbar_4}$, and $p_{3(4)}^{\hbar_1}$ are requesting to the most preferred doctor from their respective preference list $i.e.$ $d_{3(4)}^{\mathcal{H}_2}$, $d_{3(3)}^{\mathcal{H}_4}$, $d_{3(4)}^{\mathcal{H}_2}$, and $d_{3(2)}^{\mathcal{H}_1}$ respectively. In the next step, we will check if any requested doctor among $d_{3(1)}^{\mathcal{H}_3}$, $d_{3(2)}^{\mathcal{H}_1}$, $d_{3(3)}^{\mathcal{H}_4}$, and $d_{3(4)}^{\mathcal{H}_2}$ has got the multiple request from the patients in $\boldsymbol{\mathcal{P}}_3$. Now, it can be seen that, in the first iteration of TOMHECs doctor $d_{3(4)}^{\mathcal{H}_2}$ have got requests from patients $p_{3(1)}^{\hbar_2}$, and $p_{3(3)}^{\hbar_4}$. As each doctor can be assigned to only one patient, so this competitive environment between patient $p_{3(1)}^{\hbar_2}$, and $p_{3(3)}^{\hbar_4}$ can be resolved by considering the strict preference ordering of doctor $d_{3(4)}^{\mathcal{H}_2}$ over the available patients in $\boldsymbol{\mathcal{P}}_3$. From the strict preference ordering of doctor $d_{3(4)}^{\mathcal{H}_2}$ it is clear that patient $p_{3(3)}^{\hbar_4}$ is preferred over patient $p_{3(1)}^{\hbar_2}$. Hence, patient $p_{3(1)}^{\hbar_2}$ is rejected. So, for the meanwhile $p_{3(2)}^{\hbar_3}$ gets a doctor $d_{3(3)}^{\mathcal{H}_4}$, $p_{3(3)}^{\hbar_4}$ gets a doctor $d_{3(4)}^{\mathcal{H}_2}$, and $p_{3(4)}^{\hbar_1}$ gets a doctor $d_{3(2)}^{\mathcal{H}_1}$. Now, as the patient $p_{3(1)}^{\hbar_2}$ do not get his/her (henceforth his) most preferred doctor $i.e.$ $d_{3(4)}^{\mathcal{H}_2}$ from his preference list. So, he will request the second best doctor $i.e. $ $d_{3(3)}^{\mathcal{H}_4}$ from his preference list. As doctor $d_{3(3)}^{\mathcal{H}_4}$ is already been requested by $p_{3(2)}^{\hbar_3}$, the similar situation now occurs in case of doctor $d_{3(3)}^{\mathcal{H}_4}$ where patients $p_{3(1)}^{\hbar_2}$ and $p_{3(2)}^{\hbar_3}$ are simultaneously requesting to doctor $d_{3(3)}^{\mathcal{H}_4}$. Looking at the preference list of $d_{3(3)}^{\mathcal{H}_4}$, we get, patient $p_{3(1)}^{\hbar_2}$ is preferred over $p_{3(2)}^{\hbar_3}$. So, patient $p_{3(2)}^{\hbar_3}$ is rejected. Now, $p_{3(2)}^{\hbar_3}$ request the second best doctor $i.e.$ $d_{3(4)}^{\mathcal{H}_2}$ from his preference list. In
the similar fashion, the remaining allocation is done. The final allocation is:$\{(p_{3(1)}^{\hbar_2}, d_{3(3)}^{\mathcal{H}_4}), (p_{3(2)}^{\hbar_3}, d_{3(1)}^{\mathcal{H}_3}), (p_{3(3)}^{\hbar_4}, d_{3(4)}^{\mathcal{H}_2}), (p_{3(4)}^{\hbar_1}, d_{3(2)}^{\mathcal{H}_1})\}$.

\subsection{Several properties}
The proposed TOMHECs has several compelling properties. These properties are discussed next.
\begin{proposition}\label{p1}
The matching computed by the Gale-Shapley mechanism ~\cite{Gale_AMM_1962, Shapley_2013} results in a stable matching.
\end{proposition}

\begin{proposition}\label{p2}
A stable matching computed by Gale-Shapley mechanism ~\cite{Gale_AMM_1962, Shapley_2013} is requesting party optimal.
\end{proposition}

\begin{proposition}\label{p3}
Gale-Shapley mechanism ~\cite{Gale_AMM_1962, Shapley_2013} is truthful for the requesting party.
\end{proposition} 
Following the above mentioned propositions and motivated by ~\cite{Gale_AMM_1962, Shapley_2013} we have proved that the TOMHECs results in \emph{stable}, \emph{optimal}, and \emph{truthful} allocation when all the \emph{k} different categories are taken simultaneously.

\begin{mylemma}\label{d1}
TOMHECs results in a stable allocation for the requesting party (patient party or doctor party).
\end{mylemma}
\begin{proof}
Fix a category $c_i \in \boldsymbol{\mathcal{C}}$. Let us suppose for the sake of contradiction there exists a \emph{blocking pair} $(p_{i(j)}^{\hbar_k}, d_{i(j)}^{\mathcal{H}_l})$ that results in an unstable matching $\boldsymbol{\mathcal{M}}$ for the requesting party. As their exists a \emph{blocking pair} $(p_{i(j)}^{\hbar_k}, d_{i(j)}^{\mathcal{H}_l})$ it may be due to the case that $(p_{i(j)}^{\hbar_k}, d_{i(k)}^{\mathcal{H}_j})$ and $(p_{i(k)}^{\hbar_j}, d_{i(j)}^{\mathcal{H}_l})$ are their in the resultant matching $\boldsymbol{\mathcal{M}}$. This situation will arise only when $d_{i(j)}^{\mathcal{H}_l} \succ_{j}^{i} d_{i(k)}^{\mathcal{H}_j}$ $i.e.$ in the strict preference ordering of patient $p_{i(j)}^{\hbar_k}$ doctor $d_{i(j)}^{\mathcal{H}_l}$ is preferred over doctor $d_{i(k)}^{\mathcal{H}_j}$. From the matching result $\boldsymbol{\mathcal{M}}$ obtained, it can be seen that in-spite the fact that $d_{i(j)}^{\mathcal{H}_l} \succ_{j}^{i} d_{i(k)}^{\mathcal{H}_j}$; $d_{i(j)}^{\mathcal{H}_l}$ is not matched with $p_{i(j)}^{\hbar_k}$ by the TOMHECs. So, this upset may happen only when doctor $d_{i(j)}^{\mathcal{H}_l}$ received a proposal from a patient $p_{i(k)}^{\hbar_j}$ to whom $d_{i(j)}^{\mathcal{H}_l}$ prefers over $p_{i(j)}^{\hbar_k}$ $i.e.$ $p_{i(k)}^{\hbar_j} \curlyeqsucc_{i}^{j} p_{i(j)}^{\hbar_k}$. Hence, this contradicts the fact that the $(p_{i(j)}^{\hbar_k}, d_{i(j)}^{\mathcal{H}_l})$ is a \emph{blocking pair}. As their exists no \emph{blocking pair}, it can be said that the resultant matching by TOMHECs is \emph{stable}.\\
From our claim that, the TOMHECs results in a \emph{stable} matching in a particular category $c_i$, it must be true for any category. Hence, it must be true for the system considering the \emph{k} categories simultaneously.  
\end{proof}
\begin{mylemma}\label{d2}
A stable allocation resulted by TOMHECs is requesting party (patient or doctor) optimal.
\end{mylemma}
\begin{proof}
Fix a category $c_i$. Let us suppose for the sake of contradiction that the allocation set $\boldsymbol{\mathcal{M}}$ obtained using TOMHECs is not an \emph{optimal} allocation for requesting party (say \emph{patient party}). Then, from \emph{Lemma} \ref{d1} there 
exists a \emph{stable} allocation $\boldsymbol{\mathcal{M}'}$ such that $\boldsymbol{\mathcal{M}'}(p_{i(j)}^{\hbar_k}) \succ_{j}^{i} \boldsymbol{\mathcal{M}}(p_{i(j)}^{\hbar_k})$ or $\boldsymbol{\mathcal{M}'}(p_{i(j)}^{\hbar_k}) =_{j}^{i} \boldsymbol{\mathcal{M}}(p_{i(j)}^{\hbar_k})$ for at least one patient $p_{i(j)}^{\hbar_k}$ $\in$ $\boldsymbol{\mathcal{P}}_i$. Therefore, it must be the case that, some patient $p_{i(j)}^{\hbar_k}$ proposes to $\boldsymbol{\mathcal{M}'}(p_{i(j)}^{\hbar_k})$ before $\boldsymbol{\mathcal{M}}(p_{i(j)}^{\hbar_k})$ since $\boldsymbol{\mathcal{M}'}(p_{i(j)}^{\hbar_k}) \succ_{i}^{j} \boldsymbol{\mathcal{M}}(p_{i(j)}^{\hbar_k})$ and is rejected by $\boldsymbol{\mathcal{M}'}(p_{i(j)}^{\hbar_k})$. Since doctor $\boldsymbol{\mathcal{M}'}(p_{i(j)}^{\hbar_k})$ rejects patient $p_{i(j)}^{\hbar_k}$, the doctor $\boldsymbol{\mathcal{M}'}(p_{i(j)}^{\hbar_k})$ must have received a better proposal from a patient $p_{i(k)}^{\hbar_j}$ to whom doctor $\boldsymbol{\mathcal{M}'}(p_{i(j)}^{\hbar_k})$ prefers over $p_{i(j)}^{\hbar_k}$ $i.e.$ $p_{i(k)}^{\hbar_j} \curlyeqsucc_{i}^{j} p_{i(j)}^{\hbar_k}$. Since, this is the first iteration at which a doctor rejects a patient under $\boldsymbol{\mathcal{M}'}$. It follows that the allocation $\boldsymbol{\mathcal{M}}$ is preferred over allocation $\boldsymbol{\mathcal{M}'}$ for the patient $p_{i(j)}^{\hbar_k}$. Hence, this contradicts the fact that the allocation set $\boldsymbol{\mathcal{M}}$ obtained using TOMHECs is not an \emph{optimal} allocation. As their exists an optimal allocation $\boldsymbol{\mathcal{M}}$.\\
Form our claim that, the TOMHECs results in \emph{optimal} allocation in a particular category $c_i$, it must be true for any category. Hence, it must be true for the system considering the \emph{k} categories simultaneously.   
\end{proof}
\begin{mylemma}\label{d3}
A stable allocation resulted by TOMHECs is requesting party (patient or doctor) truthful.
\end{mylemma}
\begin{proof}
Fix a category $c_i$. Let us suppose for the sake of contradiction that the matching set $\boldsymbol{\mathcal{M}}$ obtained using TOMHECs is not a \emph{truthful} allocation for requesting party (say \emph{patient party}). The TOMHECs results in \emph{stable} matching $\boldsymbol{\mathcal{M}}$ when all the members of the proposing party reports their true preferences. Now, let's say a patient $p_{i(j)}^{\hbar_k}$ misreport his preference list $\succ_{j}^{i}$ and getting better off in the resultant matching $\boldsymbol{\mathcal{M}}'$. Let $\mathcal{P}'_i$ be the set of patients who are getting better off in $\boldsymbol{\mathcal{M}}'$ as against $\boldsymbol{\mathcal{M}}$. Let $\mathcal{D}'_i$ be the set of doctors matched to patients in $\mathcal{P}'_i$ in matching $\boldsymbol{\mathcal{M}}'$. Let $d_{i(k)}^{\mathcal{H}_\ell}$ be the doctor that $p_{i(j)}^{\hbar_k}$ gets in $\boldsymbol{\mathcal{M}}'$. Since $\boldsymbol{\mathcal{M}}$ is \emph{stable}, we know that $d_{i(k)}^{\mathcal{H}_{\ell}}$ cannot prefer $p_{i(j)}^{\hbar_k}$ to the patient got in $\boldsymbol{\mathcal{M}}$, because this would make ($p_{i(j)}^{\hbar_k}$, $d_{i(k)}^{\mathcal{H}_\ell}$) a \emph{blocking pair} in $\boldsymbol{\mathcal{M}}$ (see \textbf{Lemma 1}). In other words, doctor $\boldsymbol{\mathcal{M}}(d_{i(k)}^{\mathcal{H}_\ell})$ $\curlyeqsucc_{i}^{k}$ $p_{i(j)}^{\hbar_k}$. Now, if $\boldsymbol{\mathcal{M}}(d_{i(k)}^{\mathcal{H}_\ell})$ patient would not improve in $\boldsymbol{\mathcal{M}}'$ then $\boldsymbol{\mathcal{M}}(d_{i(k)}^{\mathcal{H}_\ell})$ $\curlyeqsucc_{i}^{k}$ $p_{i(j)}^{\hbar_k}$. Hence, $d_{i(k)}^{\mathcal{H}_\ell}$ can not be matched with $p_{i(j)}^{\hbar_k}$ in $\boldsymbol{\mathcal{M}}'$, a contradiction. Therefore, patient in $\boldsymbol{\mathcal{M}}$ also improves in $\boldsymbol{\mathcal{M}}'$. That is, $\mathcal{D}'_i$ is not the only set of doctors in $\boldsymbol{\mathcal{M}}'$ of those patient who are getting better off in $\boldsymbol{\mathcal{M}}$; but also the set of doctors where patient in  $\boldsymbol{\mathcal{M}}$ improve in $\boldsymbol{\mathcal{M}}'$. In other words, each doctor in $\mathcal{D}_i$ is matched to two different patient from $\mathcal{P}_i$ in match $\boldsymbol{\mathcal{M}}$ and $\boldsymbol{\mathcal{M}}'$, being better off in $\boldsymbol{\mathcal{M}}$ than in $\boldsymbol{\mathcal{M}}'$. It can also be proved using $\textbf{Lemma 1}$ that $\boldsymbol{\mathcal{M}}'$ is not stable; a contradiction that terminates the proof.\\
From our claim that, the TOMHECs results in a \emph{truthful} matching in a particular category $c_i$, it must be true for any category. Hence, it must be true for the system considering the \emph{k} categories simultaneously.
\end{proof}

\section{Further analytics-based analysis}
In order to provide sufficient reasoning to our simulation results presented in section 6, the two proposed mechanisms are in general analyzed on the ground of the expected distance of allocation done by the mechanisms from the top most preference. As a warm up, first the the analysis is done for any patient \emph{j}, to estimate the expected distance of allocation from the top most preference. After that the analysis is extended to more general setting where all the patients present in the system are considered. It is to be noted that the results revealed by the simulations can easily be verified by the lemmas below.   
\begin{mylemma}\label{d0}
The allocation resulted by RAMHECs for any patient (or doctor) j being considered first is on an average $\frac{n}{2}$ distance away from its most preferred doctor (or patient) $i.e.~ E[Z] \simeq \frac{n}{2}$; where $Z$ is the random variable measuring the distance from the top most preference.
\end{mylemma}
\begin{proof}
Fix a category $c_i \in \boldsymbol{\mathcal{C}}$, and an arbitrary patient \emph{j} being considered first. In RAMHECs, for any arbitrary patient (AP) being considered first are allotted a random doctor from his preference list. The index position of the doctor in the preference list is decided by \emph{k}, where $k = 1, 2, \ldots, n$. Now, when a doctor is selected randomly from the preference list any of these \emph{k} ($1 \leq k \leq n$) may be selected. So any index \emph{k} could be the outcome of the experiment (allocation of a doctor) and it is to be noted that selection of any such \emph{k} is equally likely. Therefore, for each \emph{k} such that $1 \leq k \leq n$ any $k^{th}$ doctor can be selected with probability $\frac{1}{n}$. For $k = 1, 2, \ldots, n$, we define indicator random variable $X_k$ where
\begin{equation*}
X_k = I\{k^{th}~ doctor~ selected~ from~ patients'~ preference~ list\}
\end{equation*}   
\begin{equation*}
X_k =
\begin{cases}
 1, & \textit{if $k^{th}$ doctor is selected} \\
   0,         & \textit{otherwise}
 \end{cases}
\end{equation*}   
Taking expectation both side, we get;
\begin{equation*}
E[X_k] = E[I\{k^{th}~ doctor~ selected~ from~ patients'~ preference~ list\}]
\end{equation*} 
As always with the \emph{indicator random variable}, the expectation is just the probability of the corresponding event \cite{Coreman_2009}:  
\begin{equation*}
 E[X_k]= 1 \cdot Pr\{X_k=1\} + 0 \cdot Pr\{X_k=0\} = 1 \cdot Pr\{X_k=1\} = \frac{1}{n}
\end{equation*}
For a given call to RAMHECs, the indicator random variable $X_k$ has the value 1 for exactly one value of \emph{k}, and it is 0 for all other \emph{k}. For $X_k=1$, we can measure the distance of $k^{th}$ allocated doctor from the most preferred doctor in the patient \emph{j}'s preference list. So, let $d_k$ be the distance of $k^{th}$ allocation from the best preference. More formally, it can be represented in the case analytic form as:
\begin{align*}
Z = 
\begin{cases}
d_0: &\textit{If $1^{st}$ agent is selected from the preference list}~ (k=1)\\
d_1: &\textit{If $2^{nd}$ agent is selected from the preference list}~ (k=2)\\
\mathrel{\makebox[\widthof{=}]{\vdots}}    & \mathrel{\makebox[\widthof{=}]{\vdots}} \\
d_{n-1}: & \textit{If $n^{th}$ agent is selected from the preference list}~ (k=n)
\end{cases}
\end{align*}
Where $Z$ is the random variable measuring the distance of the allocation from the patient's top most preference. Here, $d_0=0$, $d_1=1$, $d_2=2$, $\ldots$, $d_{n-1}=n-1$. 
It is to be observed that, once the doctor \emph{k} is selected from the patient \emph{j}'s preference list, the value calculation of $d_k$ is no way dependent on \emph{k}. Now,   
observe that the random variable $Z$ that we really care about can be formulated as:
\begin{equation*}
Z = \sum_{k=1}^{n}X_k \cdot d_{k-1} 
\end{equation*}
Taking expectation both side. We get;
\begin{equation*}
\hspace*{-35mm}E[Z] = E\Bigg[\sum_{k=1}^{n}X_k \cdot d_{k-1}\Bigg] 
\end{equation*}
\begin{equation*}
\hspace*{38mm} = \sum_{k=1}^{n}E[X_k \cdot d_{k-1}]           ~~~~~~~~~~(by ~linearity ~of ~expectation)
\end{equation*}
\begin{equation*}
\hspace*{38mm} = \sum_{k=1}^{n}E[X_k] \cdot E[d_{k-1}]           ~~~~~~~~~~(X_k ~and ~d_{k-1} ~are ~independent)
\end{equation*}
\begin{equation*}
= \sum_{k=1}^{n}\frac{1}{n} \cdot E[d_{k-1}] = \frac{1}{n} \sum_{k=1}^{n} E[d_{k-1}] 
\end{equation*}
\begin{equation*}
\hspace*{39mm} = \frac{1}{n} \sum_{k=1}^{n} d_{k-1} ~~~~(once ~\emph{k}~ is~ fixed ~ d_{k-1} ~becomes ~constant)
\end{equation*}
\begin{equation*}
\hspace*{-28mm} = \frac{1}{n} \cdot \frac{(n-1)(n)}{2}
\end{equation*}
\begin{equation*}
\hspace*{-28mm} =  \frac{(n-1)}{2} \simeq \frac{n}{2},
\end{equation*}
as claimed.
\end{proof}

\begin{mylemma}\label{d00}
In RAMHECs, $E[D] \simeq \frac{n^2}{16}$; where $D$ is the total distance of all the patients in the system from the top most preference.
\end{mylemma}
\begin{proof}
Fix a category $c_i \in \boldsymbol{\mathcal{C}}$. We are analyzing, the expected distance of the allocations done to the patients by RAMHECs from the top most preferences. For this purpose, as there are \emph{n} patients, the index of these patients are captured by \emph{i} such that $i = 1, 2, \ldots, n$. Without loss of generality, the patients are considered in some order. The index position of the doctor in any patient \emph{j}'s preference list is decided by \emph{k}, where $k = 1, 2, \ldots, n$. For any patient \emph{i} ($1\leq i\leq n$) selected first, when a doctor is selected randomly from the preference list any of the available $k$ ($1\leq k \leq n$) doctors can be selected. So, any index \emph{k} could be the outcome of the experiment (allocation of doctor) and any such \emph{k} is equally likely. But what could the case, if instead of considering the patient in the first place, say a patient is selected in $i^{th}$ iteration. In that case, from the construction of RAMHECs the length of the preference list of the patient under consideration would be $n-i+1$. So, when a doctor is selected randomly from the preference list, any of the $(n-i+1)$ doctors may be selected. It is to be noted that the selection of any of the $(n-i+1)$ doctors is equally likely. Therefore, for a patient under consideration in $i^{th}$ iteration, for each \emph{k} such that $1 \leq k \leq n-i+1$ any $k^{th}$ doctor can be selected with probability $\frac{1}{n-i+1}$. Here, we are assuming that each agent's top preferences are still remaining when that agent is considered by the RAMHECs. To get the lower bound this is the best possible setting. If an agent is not provided that list, he will be further away from his top most preference. For each patient \emph{i} and for $k = 1, 2, \ldots, n$, we define indicator random variable $X_{ik}$ where
\begin{equation*}
X_{ik} = I\{k^{th}~ doctor~ selected~ from~ patient ~i's~ preference~ list\}
\end{equation*}   
\begin{equation*}
X_{ik} =
\begin{cases}
 1, & \textit{if $k^{th}$ doctor is selected from patient i's preference list} \\
   0,         & \textit{otherwise}
 \end{cases}
\end{equation*}   
Taking expectation both side, we get;
\begin{equation*}
E[X_{ik}] = E[I\{k^{th}~ doctor~ selected~ from~ patient ~i~ preference~ list\}]
\end{equation*} 
As always with the \emph{indicator random variable}, the expectation is just the probability of the corresponding event:  
\begin{equation*}
 E[X_{ik}]= 1 \cdot Pr\{X_{ik}=1\} + 0 \cdot Pr\{X_{ik}=0\} = 1 \cdot Pr\{X_{ik}=1\} = \frac{1}{n-i+1}
\end{equation*}
For a given call to RAMHECs, the indicator random variable $X_{ik}$ has the value 1 for exactly one value of \emph{k}, and it is 0 for all other \emph{k}. For $X_{ik}=1$, we can measure the distance of $k^{th}$ allocated doctor from the most preferred doctor in the patient \emph{j}'s preference list. So, let $d_{ik}$ be the distance of $k^{th}$ allocation from the best preference. More formally, it can be represented in the case analytic form as:
\begin{align*}
D = 
\begin{cases}
d_{i0}: &\textit{If $1^{st}$ agent is selected from the preference list}~ (k=1)\\
d_{i1}: &\textit{If $2^{nd}$ agent is selected from the preference list} ~(k=2)\\
\mathrel{\makebox[\widthof{=}]{\vdots}}    & \mathrel{\makebox[\widthof{=}]{\vdots}} \\
d_{i(n-1)}: & \textit{If $n^{th}$ agent is selected from the preference list~ (k=n)}
\end{cases}
\end{align*}
Where $D$ is the total distance of all the patients in the system from the top most preference.
It is to be observed that, once the doctor \emph{k} is selected from the patient \emph{j}'s preference list, the value calculation of $d_k$ is no way dependent on \emph{k}. Now,   
observe that the random variable $D$ that we really care about is given as:
\begin{equation*}
\hspace*{-50mm}D \geq \sum_{i=1}^{n}\sum_{k=1}^{n-i+1}X_{ik} \cdot d_{ik} 
\end{equation*}
Taking expectation both side. We get;
\begin{equation*}
\hspace*{-45mm}E[D] \geq E[\sum_{i=1}^{n}\sum_{k=1}^{n-i+1}X_{ik} \cdot d_{ik}] 
\end{equation*}
\begin{equation*}
\hspace*{31mm} = \sum_{i=1}^{n}\sum_{k=1}^{n-i+1}E[X_{ik} \cdot d_{ik}]           ~~~~~~~~~~(by ~linearity ~of ~expectation)
\end{equation*}
\begin{equation*}
\hspace*{31mm} = \sum_{i=1}^{n}\sum_{k=1}^{n-i+1}E[X_{ik}] \cdot E[d_{ik}]           ~~~~~~~~~~(X_{ik} ~and ~d_{ik} ~are ~independent)
\end{equation*}
\begin{equation*}
\hspace*{-28mm}= \sum_{i=1}^{n}\sum_{k=1}^{n-i+1}\frac{1}{n-i+1} \cdot d_{ik} 
\end{equation*}
\begin{equation*}
\hspace*{-40mm} \geq \sum_{i=1}^{n}\sum_{k=1}^{n-i+1} \frac{1}{n} \cdot d_{ik} 
\end{equation*}
\begin{equation*}
\hspace*{-41mm} = \frac{1}{n} \sum_{i=1}^{n}\sum_{k=1}^{n-i+1}  d_{ik}
\end{equation*}
\begin{equation*}
\hspace*{-09mm} =  \frac{1}{n} \Bigg[\sum_{i=1}^{\frac{n}{2}}\sum_{k=1}^{n-i+1}  d_{ik} + \sum_{i=\frac{n}{2}}^{n}\sum_{k=1}^{n-i+1}  d_{ik}\Bigg] 
\end{equation*}
\begin{equation*}
\hspace*{05mm} \geq  \frac{1}{n} \Bigg[\sum_{i=1}^{\frac{n}{2}}\sum_{k=1}^{n-i+1}  d_{ik}\Bigg] + \Bigg[\sum_{i=\frac{n}{2}}^{n}\sum_{k=1}^{n-i+1}  0\Bigg] ~~~(\footnotemark \footnotetext{discarding~ the~ lower~ order~ terms})
\end{equation*}
\begin{equation*}
\hspace*{-35mm} \geq  \frac{1}{n} \Bigg[\sum_{i=1}^{\frac{n}{2}}\sum_{k=\frac{n}{2}}^{n-i+1}  d_{ik}\Bigg]  
\end{equation*}
\begin{equation*}
\hspace*{-25mm} \geq  \frac{1}{n} \Bigg[\sum_{i=1}^{\frac{n}{2}}\sum_{k=\frac{n}{2}}^{n-i+1}  d_{i\frac{n}{2}}\Bigg]  ~~~(\footnotemark \footnotetext{replacing~ each~ term~ of~ the~ series~ by~ its~ first~ term})
\end{equation*}
\begin{equation*}
\hspace*{-35mm}  = \frac{1}{n} \Bigg[\sum_{i=1}^{\frac{n}{2}}\sum_{k=\frac{n}{2}}^{n-i+1} \frac{n}{2}\Bigg] 
\end{equation*}
\begin{equation*}
\hspace*{-35mm}= \frac{1}{2} \Bigg[\sum_{i=1}^{\frac{n}{2}}\sum_{k=\frac{n}{2}}^{n-i+1} 1\Bigg] 
\end{equation*}
\begin{equation*}
\hspace*{-35mm} \geq \Bigg(\frac{1}{2} \sum_{j=1}^{\frac{n}{2}} j\Bigg) -1 
\end{equation*}
\begin{equation*}
\hspace*{-33mm} =  \frac{1}{2}\Bigg[\frac{\frac{n}{2}(\frac{n}{2}+1)}{2}\Bigg]-1
\end{equation*}
\begin{equation*}
\hspace*{-30mm} =  \frac{n^2+2n-16}{16} \simeq \frac{n^2}{16}
\end{equation*}
as claimed. It is to be observed that for each agent, the expected distance of allocation done by RAMHECs from the top preference in an amortized sense is $\frac{n}{16}$. 
\end{proof}

\begin{mylemma}\label{d01}
The expected number of rejections for any arbitrary patient (or doctor) j resulted by TOMHECs is constant. If the probability of any k length rejection is considered as $\frac{1}{2}$ $i.e.$ $Pr\{Y_k=1\}= \frac{1}{2}$ then $E[Y] = 2$; where $Y$ is the random variable measuring the total number of rejections made to the patient (or doctor) under consideration. 
\end{mylemma}
\begin{proof}
Fix a category $c_i \in \boldsymbol{\mathcal{C}}$, and an arbitrary patient \emph{j}. To analyze the expected number of rejections suffered by the patient under consideration in case of TOMHECs, we capture the total number of rejections done to any patient \emph{j} by a random variable $Y$. So, the expected number of rejections suffered by any patient \emph{j} is given as $E[Y]$. It is considered that the rejection by any member $k = 0, \ldots, n-1$, present on the patients' \emph{j} preference list is an independent experiment. It means that, the \emph{m} length rejections suffered by an arbitrary patient \emph{j} is no way dependent on any of the previous $m-1$ rejections. Let us suppose for each $0 \leq k \leq n-1 $, the probability of rejection by any $k^{th}$ doctor be $\frac{1}{2}$ (it can be any value between 0 and 1 depending on the scenario). For $k=0,\ldots, n-1$, we define indicator random variable $Y_k$ where
\begin{equation*}
Y_k = I\{k~length~rejection\}
\end{equation*}   
\begin{equation*}
Y_k =
\begin{cases}
 1, & \textit{if k length rejection} \\
   0,         & \textit{otherwise}
 \end{cases}
\end{equation*}   
Taking expectation both side, we get;
\begin{equation*}
E[Y_k] = E[I\{k~length~rejection\}]
\end{equation*} 
As always with the \emph{indicator random variable}, the expectation is just the probability of the corresponding event:  
\begin{equation*}
 E[Y_k]= 1 \cdot Pr\{Y_k=1\} + 0 \cdot Pr\{Y_k=0\} = 1 \cdot Pr\{Y_k=1\} = \Bigg(\frac{1}{2}\Bigg)^k
\end{equation*}

Observe that the random variable $Y$ that we really care about is given as,
\begin{equation*}
Y = \sum_{k=0}^{n-1}Y_k 
\end{equation*}
Taking expectation both side. We get;
\begin{equation*}
E[Y] = E\Bigg[\sum_{k=0}^{n-1}Y_k \Bigg] 
\end{equation*}
\begin{equation*}
\hspace*{62mm}= \sum_{k=0}^{n-1}E[Y_k] ~~~~~~~(by ~linearity ~of ~expectation)
\end{equation*}
\begin{equation*}
\hspace*{33mm}= \sum_{k=0}^{n-1}\Bigg(\frac{1}{2}\Bigg)^{k} < \sum_{k=0}^{\infty}\Bigg(\frac{1}{2}\Bigg)^{k}
\end{equation*}

\begin{equation*}
\hspace*{13mm} = \frac{1}{1-(\frac{1}{2})} = 2
\end{equation*}
as claimed. Moreover, if we consider the probability of $k^{th}$ rejection as $\frac{2}{3}$ then, the expected number of rejections will be given as 3 $i.e$ $E[Y]=3$. Similarly, $E[Y]=10$ if the probability of $k^{th}$ rejection is taken as $\frac{9}{10}$. It means that, even with the high probability of rejection to any arbitrary patient \emph{j} by the members of the proposed party, there is a chance that after constant number of rejections patient \emph{j} will be allocated a good doctor according to his choice. Hence, we can say that each agent's allocation is not far away from his top most preference.   
\end{proof}

\begin{mylemma}\label{d02}
In TOMHECs, $E[R] = 2n$, where $R$ is the random variable measuring the total number of rejections made to all the patients.
\end{mylemma}
\begin{proof}
Fix a category $c_i \in \boldsymbol{\mathcal{C}}$. We are analyzing the total number of rejections suffered by all the patients in expectation. For this purpose, as there are \emph{n} patients, the index of these patients are captured by \emph{i} such that $i = 1, 2, \ldots, n$. The index position of the doctor in any patient $j's$ preference list is decided by \emph{k}, where $k=1, 2, \ldots, n$. We capture the total number of rejections done to all patients by a random variable $R$. So, the expected number of rejections suffered by all the patients is given as $E[R]$. It is considered that the rejection by any member $k = 1, \ldots, n-1$, present on the patients' \emph{i} preference list is an independent experiment. It means that, the \emph{m} length rejections suffered by an arbitrary patient \emph{i} is no way dependent on any of the previous $m-1$ rejections. Let us suppose for each patient \emph{i} and for each $1 \leq k \leq n-1 $, the probability of rejection by any $k^{th}$ doctor be $\frac{1}{2}$ (it can be any value between 0 and 1 depending on the scenario). For $k=1,\ldots, n-1$, we define indicator random variable $R_{ik}$ where
\begin{equation*}
R_{ik} = I\{k~length~rejection~of i^{th} ~patient\}
\end{equation*}   
\begin{equation*}
R_{ik} =
\begin{cases}
 1, & \textit{if k length rejection of $i^{th}$ patient} \\
   0,         & \textit{otherwise}
 \end{cases}
\end{equation*}   
Taking expectation both side, we get;
\begin{equation*}
E[R_{ik}] = E[I\{k~length~rejection~of~ i^{th}~ patient\}]
\end{equation*} 
As always with the \emph{indicator random variable}, the expectation is just the probability of the corresponding event:  
\begin{equation*}
 E[R_{ik}]= 1 \cdot Pr\{R_{ik}=1\} + 0 \cdot Pr\{R_{ik}=0\} = 1 \cdot Pr\{R_{ik}=1\} = \Bigg(\frac{1}{2}\Bigg)^k
\end{equation*}

Observe that the random variable $R$ that we really care about is given as,
\begin{equation*}
R = \sum_{i=1}^{n}\sum_{k=1}^{n-i}R_{ik} 
\end{equation*}
Taking expectation both side. We get;
\begin{equation*}
E[R] = E\Bigg[\sum_{i=1}^{n}\sum_{k=1}^{n-i}R_{ik} \Bigg] 
\end{equation*}
\begin{equation*}
\hspace*{62mm}= \sum_{i=1}^{n}\sum_{k=1}^{n-i}E[R_{ik}] ~~~~~~~(by ~linearity ~of ~expectation)
\end{equation*}
\begin{equation*}
\hspace*{17mm}= \sum_{i=1}^{n}\sum_{k=1}^{n-i}\Bigg(\frac{1}{2}\Bigg)^{k} 
\end{equation*}
\begin{equation*}
\hspace*{17mm} < \sum_{i=1}^{n}\sum_{k=0}^{\infty}\Bigg(\frac{1}{2}\Bigg)^{k}
\end{equation*}
\begin{equation*}
\hspace*{13mm} = \sum_{k=0}^{n-1}\frac{1}{1-(\frac{1}{2})} = 2n
\end{equation*}
as claimed.      
\end{proof}

\begin{corollary}
It is to be observed that for each patient, the expected number of rejections in case of TOMHECs in an amortized sense is $O(1)$. As we have shown that for all n agents, the expected number of rejection are $O(n)$.
\end{corollary}
\section{Experimental findings}
The experiments are carried out in this section to compare the efficacy of the TOMHECs based on the preference lists of the doctors and patients generated randomly using \emph{Random} library in Python. RAMHECs is considered as the benchmark mechanism.
\subsection{Simulation setup}
For creating a real world healthcare scenario we have considered 10 different categories of patients and doctors for our simulation purpose.
One of the scenarios that is taken into consideration is, say there are equal number of patients and doctors present in each of the categories along with the assumption that each of the patients are providing strict preference ordering (generated randomly) over all the available doctors and also each of the doctors are providing strict preference ordering over all the available patients. Second scenario is the case where, there are equal number of patients and doctors are present in the market. Each of the members in the respective parties are providing the strict preference ordering over the subset of the members of the opposite community. The other two scenarios $i.e.$ $m<n$ and $m>n$ with partial preference are not shown due to page limit.
\subsection{Performance metrics}
The efficacy of TOMHECs is measured under the banner of two important parameters: (a) \textbf{Satisfaction level ($\boldsymbol{\eta_{\ell}}$):} It is defined as the sum over the difference between the index of the doctor (patient) allocated from the patient's (doctor's) preference list to the index of the most preferred doctor (patient) by the patient (doctor) from his/her preference list. Considering the \emph{requesting party}, the $\boldsymbol{\eta_\ell^j}$ for $c_j$ category is defined as: $\boldsymbol{\eta_\ell^{j}} = \sum_{i=1}^{n} \bigg(\overline{\xi}_{i} - \xi_i \bigg)$; where, $\overline{\xi}_{i}$ is the index of the doctor (patient) allocated from the initially provided preference list of the patients (doctors) \emph{i}, and $\xi_i$ is the index of the most preferred doctor (patient) in the initially provided preference list of patient (doctor) \emph{i}. For \emph{k} categories, $\boldsymbol{\eta_\ell} = \sum_{j=1}^{k}\sum_{i=1}^{n} \bigg(\overline{\xi}_{i} - \xi_i \bigg)$. It is to be noted that lesser the value of satisfaction level higher will be the satisfaction of patients or doctors. (b) {\textbf{Number of preferable allocation ($\boldsymbol{\zeta}$):} The term "\emph{preferable allocation}" refers to the allocation of most preferred doctor or patient from the revealed preference lists by the patients or the doctors respectively. For a particular patient or doctor the \emph{preferable allocation} is captured by the function $f: \boldsymbol{\mathcal{P}}_{i} \rightarrow \{0, 1\}$. For the category $c_i$, the number of preferable allocation (NPA) is defined as the number of patients (doctors) getting their first choice from the initially provided preference list. So, $\boldsymbol{\zeta_i} = \sum_{j=1}^{n} f(p_{i(j)}^{\hbar_\ell})$. For \emph{k} categories $\boldsymbol{\zeta} = \sum_{i=1}^{k}\sum_{j=1}^{n} f(p_{i(j)}^{\hbar_\ell})$.

\subsection{Simulation directions} 
The three directions are seen for measuring the performance of TOMHECs, they are: (1) All the patients and doctors are reporting their true preference list. (2) When fraction of total available members  of the \emph{requesting party} are misreporting their preference lists. (3) When fraction of total available members of the \emph{requested party} are misreporting their preference lists.
\subsection{Result analysis}
 In this section, the result is simulated for the above mentioned three cases and discussed.
 \begin{table}[H]
\caption{Abbreviations used in simulation}
\label{tab:1}
\centering 
\vspace{4mm}
\scalebox{0.75}{
\begin{tabular}{|c|| l|l|}
 \hline
\textbf{Abbreviation}  & \textbf{Description}\\
 \hline
 \hline
  RAMHECs-P & Patients allocation using RAMHECs without variation.\\

 \hline
   TOMHECs-P  & Patients allocation using TOMHECs without variation.\\
 \hline
  RAMHECs-D & Doctors allocation using RAMHECs without variation.\\

 \hline
   TOMHECs-D  & Doctors allocation using TOMHECs without variation.\\ 
 \hline
   TOMHECs-PS  & Patients allocation using TOMHECs with small variation.\\
 \hline
   TOMHECs-DS  & Doctors allocation using TOMHECs with small variation.\\
 \hline
 TOMHECs-PM  & Patients allocation using TOMHECs with medium variation.\\
 \hline
 TOMHECs-DM  & Doctors allocation using TOMHECs with medium variation.\\
 \hline
 TOMHECs-PL  & Patients allocation using TOMHECs with large variation.\\
 \hline
 TOMHECs-DL  & Doctors allocation using TOMHECs with large variation.\\
 \hline
 
\end{tabular}
}
\end{table}

\paragraph*{\textbf{Expected amount of patients/doctors deviating}}
The following analysis motivated
by \cite{Coreman_2009} justifies the idea of choosing the parameters of variation. Let $\chi_{j}$ be the random variable associated with the event in which $j^{th}$ patient in $c_i$ category varies its true preference ordering.
Thus, $\chi_{j}$ = \{$j^{th}$ patient varies preference ordering\}. 
$\chi = \sum_{j =1}^{n} \chi_j$. We can write $E[\chi]$ = $\sum_{j =1}^{n} E[\chi_j]$ = $\sum_{j=1}^{n} 1/8$ = $n/8$. Here, 
Pr$\{j^{th}$ \emph{patient varies preference ordering}$\}$ is the probability that given a patient whether he will vary his 
true preference ordering. The probability of that is taken as $1/8$ (small variation).\\
Our result analysis is broadly classified into four categories:\\
\noindent $\bullet$ \textbf{Case 1a: Requesting party with full preference (FP)}
In Figure \ref{fig:sim1a} and Figure \ref{fig:sim2a}, it can be seen that the \emph{satisfaction level} of the requesting party in case of TOMHECs is more as compared to RAMHECs. As TOMHECs always allocates the most preferred member from the preference list. Under the \emph{manipulative} environment of the \emph{requesting party}, it can be seen in Figure \ref{fig:sim1a} and Figure \ref{fig:sim2a} that, the \emph{satisfaction level} of the system in case of TOMHECs with large variation is less than the \emph{satisfaction level} of the system in case of TOMHECs with medium variation is less than the \emph{satisfaction level} of the system in case of TOMHECs with small variation is less than the \emph{satisfaction level} of the system in case of TOMHECs. It is natural from the construction of TOMHECs. Considering the second parameter $i.e.$ \emph{number of preferable allocation}, it can be seen in Figure \ref{fig:sim1b} and Figure \ref{fig:sim2b} that the NPA of the requesting party in case of TOMHECs is more as compared to RAMHECs.
\begin{figure}[H]
\begin{subfigure}[b]{0.48\textwidth}
                \centering
                \includegraphics[scale=0.30]{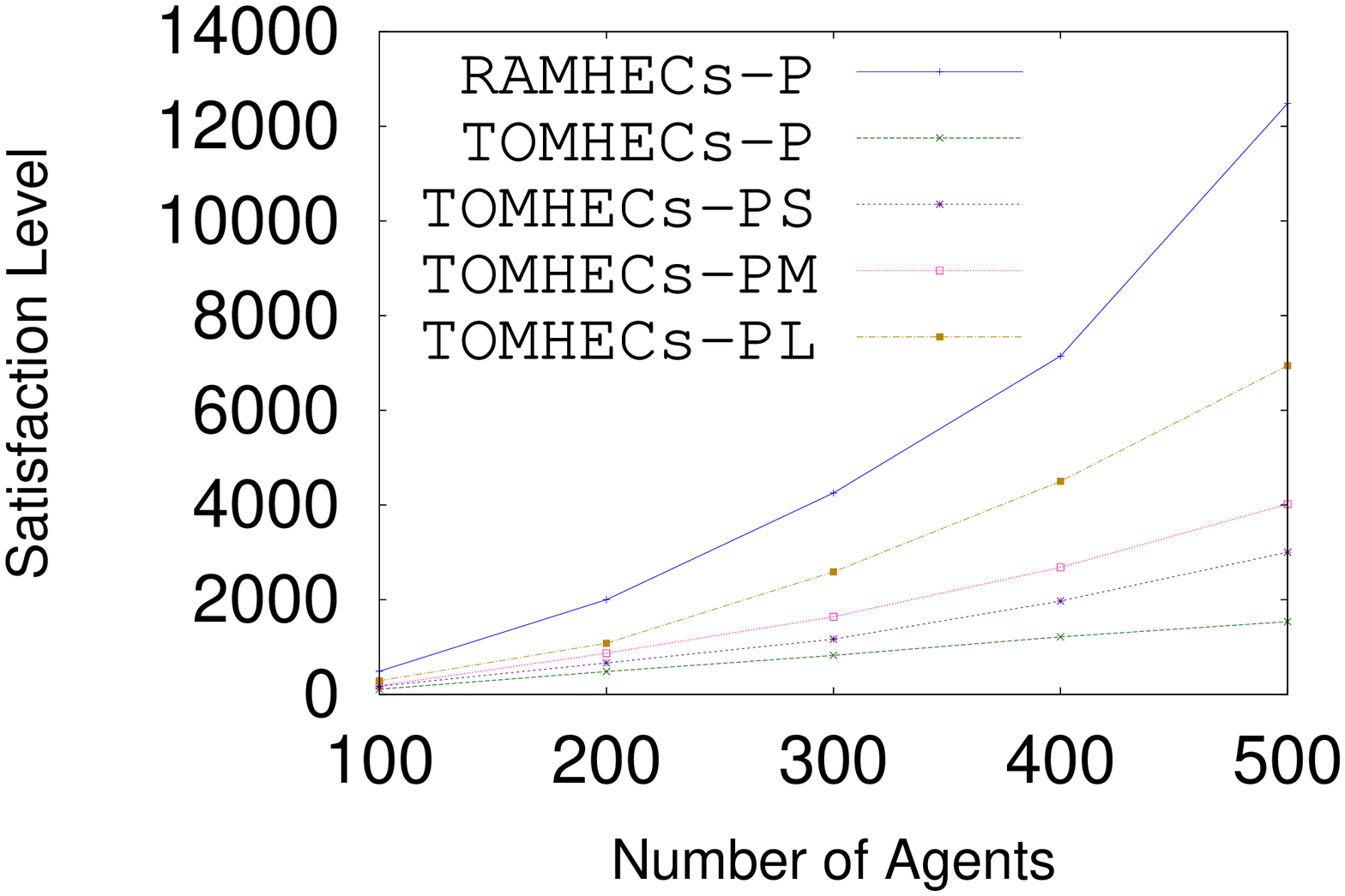}
                \subcaption{$\boldsymbol{\eta_{\ell}}$ of patients}
                \label{fig:sim1a}
        \end{subfigure}%
        \begin{subfigure}[b]{0.48\textwidth}
                \centering
                \includegraphics[scale=0.30]{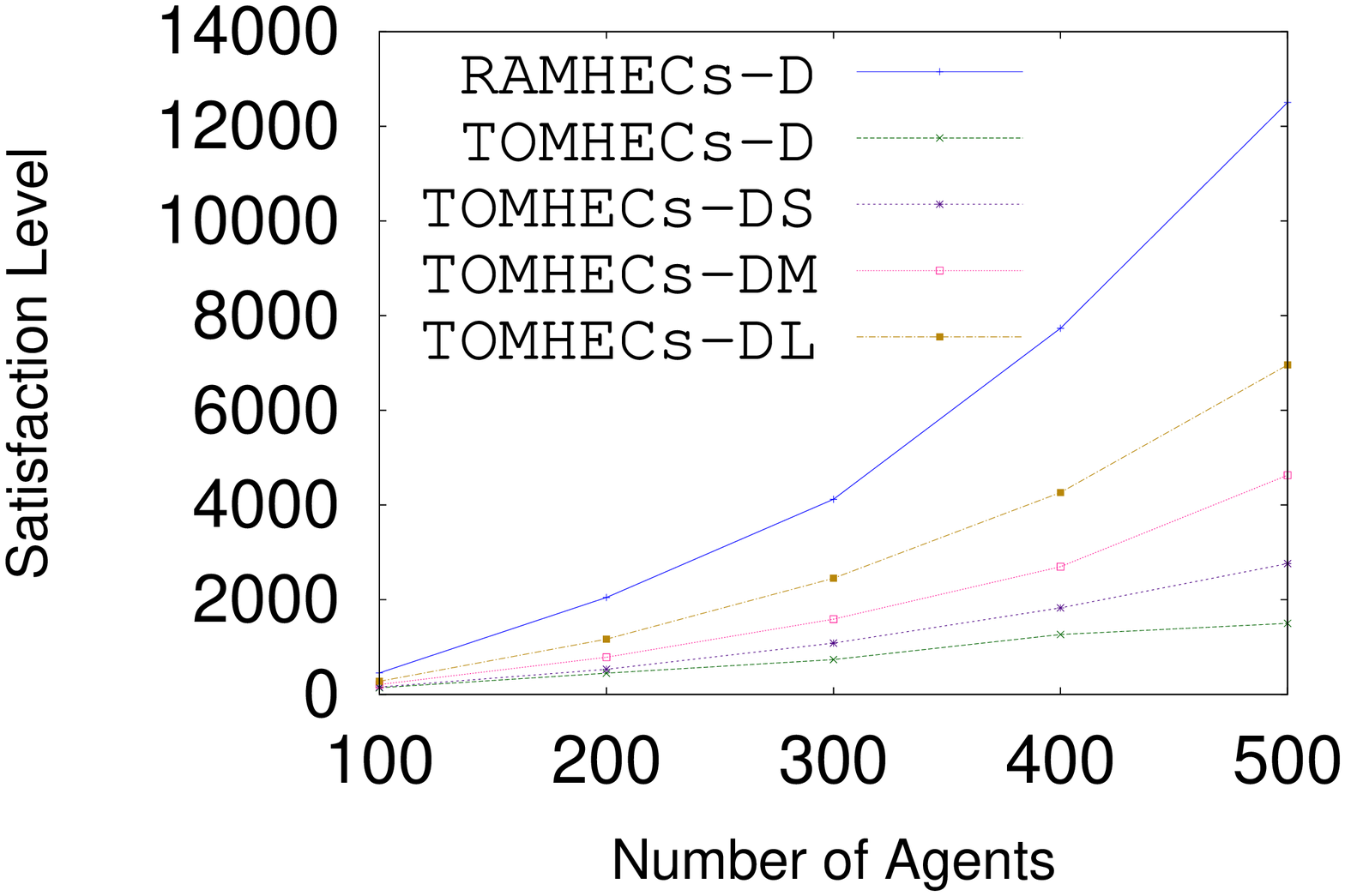}
                \subcaption{$\boldsymbol{\eta_{\ell}}$ of doctors}
                \label{fig:sim2a}
        \end{subfigure}
        \caption{$\boldsymbol{\eta_{\ell}}$ of requesting party with $m==n$}
\end{figure} 

\begin{figure}[H]
\begin{subfigure}[b]{0.48\textwidth}
                \centering
                \includegraphics[scale=0.17]{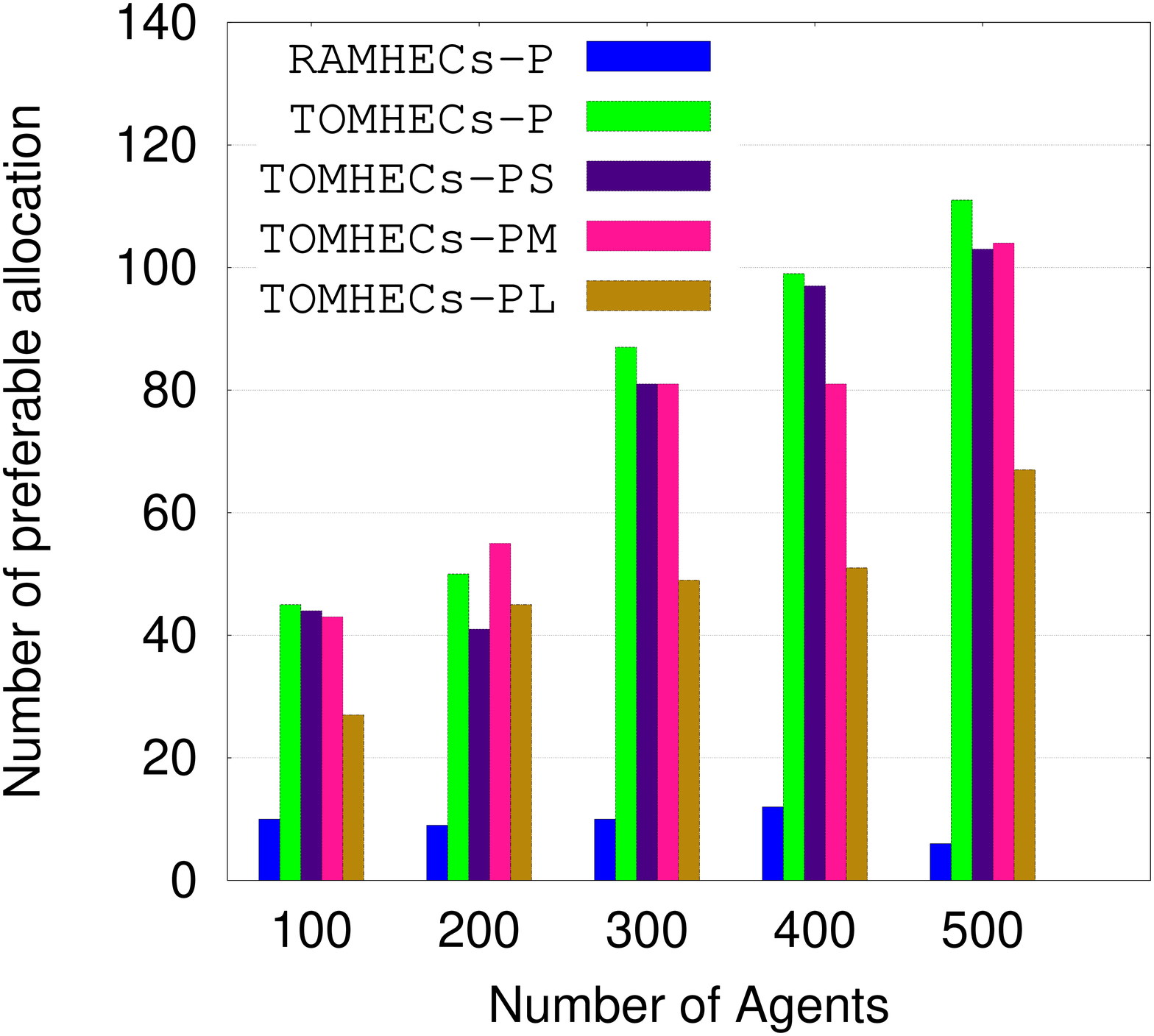}
                \subcaption{$\boldsymbol{\zeta}$ of patients}
                \label{fig:sim1b}
        \end{subfigure}%
        \begin{subfigure}[b]{0.48\textwidth}
                \centering
                \includegraphics[scale=0.17]{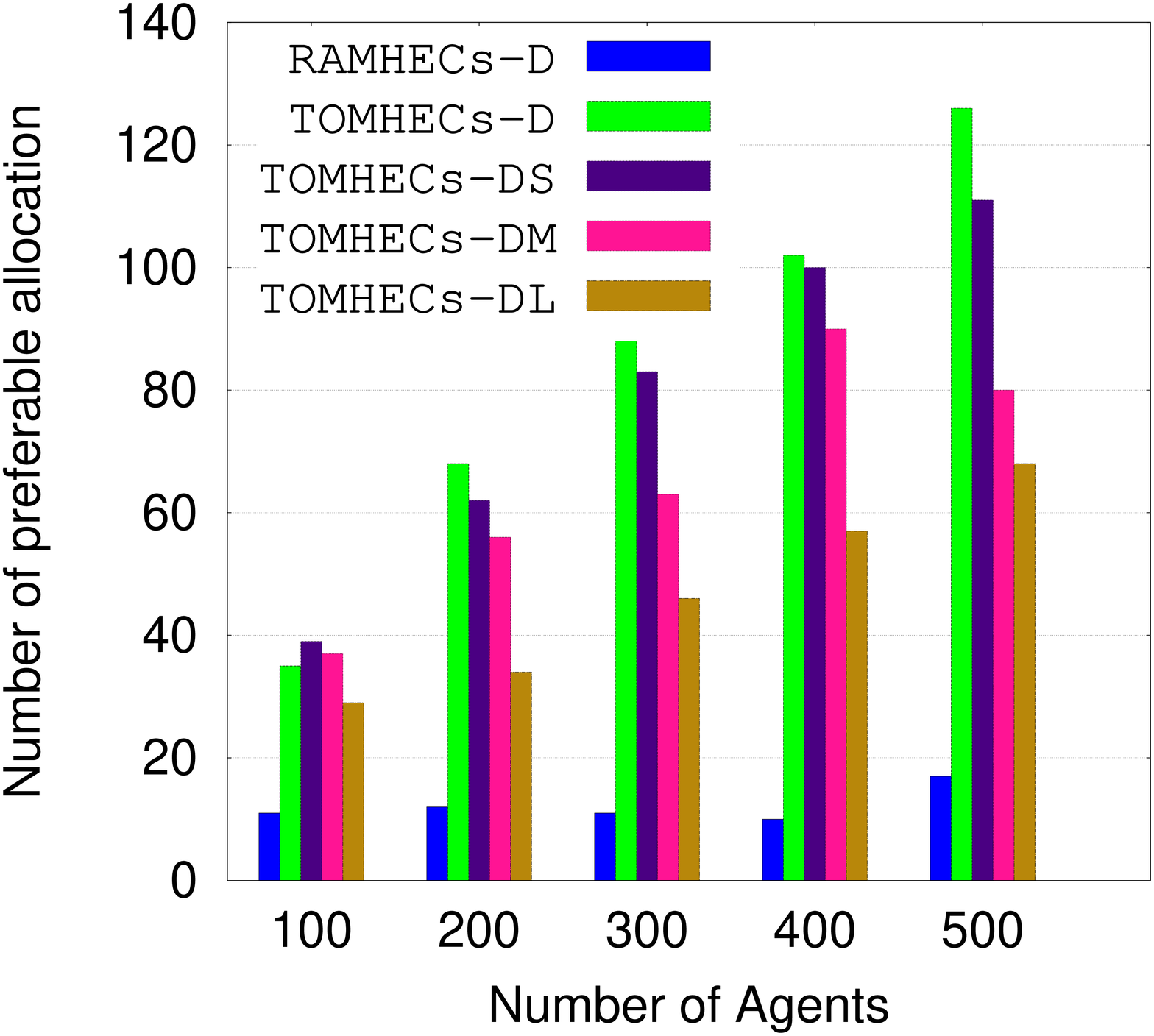}
                \subcaption{$\boldsymbol{\zeta}$ of doctors}
                \label{fig:sim2b}
        \end{subfigure}
        \caption{$\boldsymbol{\zeta}$ of requesting party with $m==n$}
\end{figure}

\noindent Under the \emph{manipulative} environment of the \emph{requesting party}, it can be seen in Figure \ref{fig:sim1b} and Figure \ref{fig:sim2b} that, the NPA of the system in case of TOMHECs with large variation is less than the NPA of the system in case of TOMHECs with medium variation is less than the NPA of the system in case of TOMHECs with small variation is less than the NPA of the system in case of TOMHECs. It is natural from the construction of TOMHECs.\\	
\noindent $\bullet$ \textbf{Case 1b: Requesting party with partial preference (PP)}
In Figure \ref{fig:sim3a} and Figure \ref{fig:sim4a}, it can be seen that the \emph{satisfaction level} of the requesting party in case of TOMHECs is more as compared to RAMHECs. As TOMHECs always allocates the most preferred member from the preference list.
\begin{figure}[H]
\begin{subfigure}[b]{0.48\textwidth}
                \centering
                \includegraphics[scale=0.30]{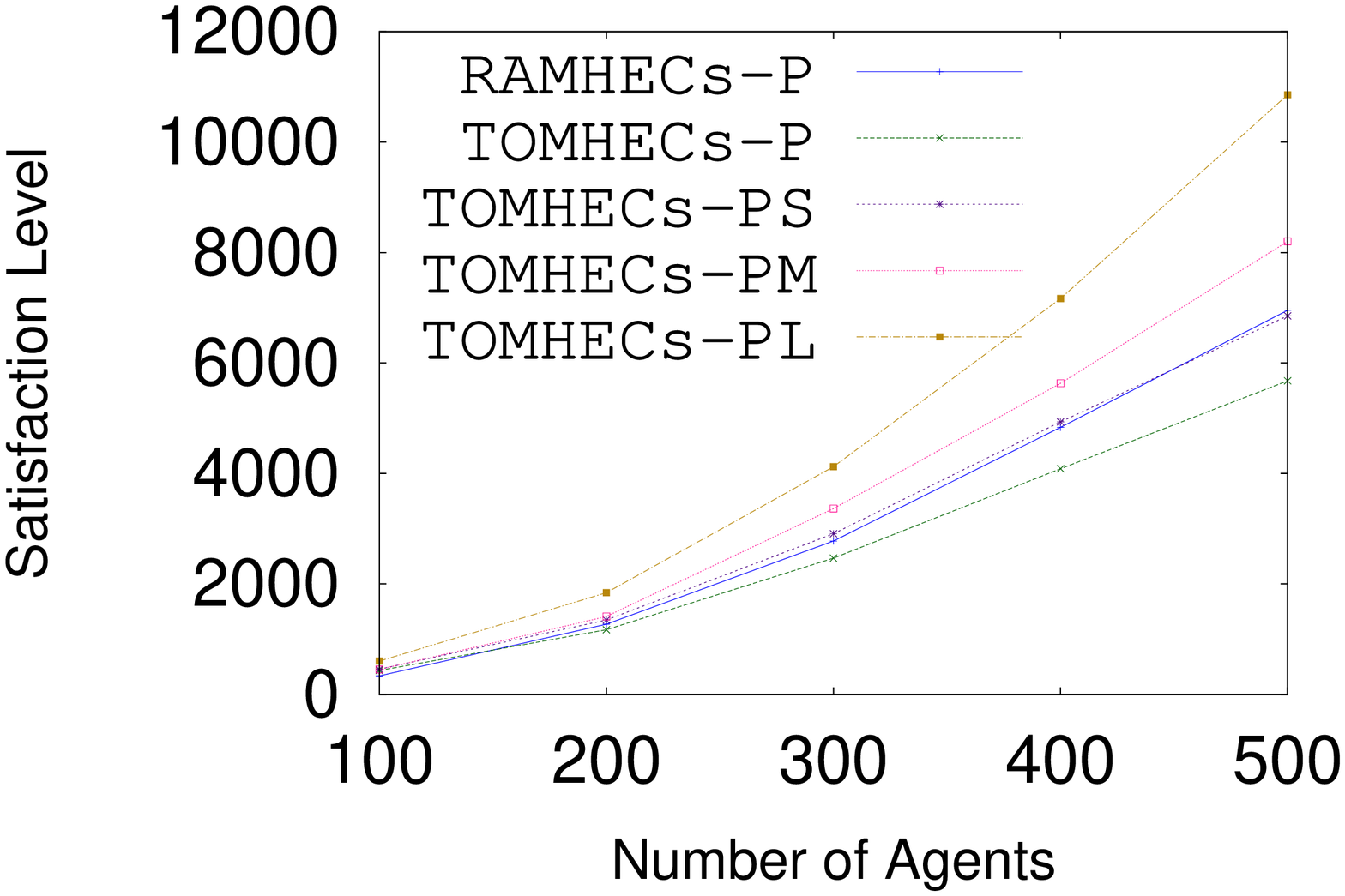}
                \subcaption{$\boldsymbol{\eta_{\ell}}$ of patients}
                \label{fig:sim3a}
        \end{subfigure}%
        \begin{subfigure}[b]{0.48\textwidth}
                \centering
                \includegraphics[scale=0.30]{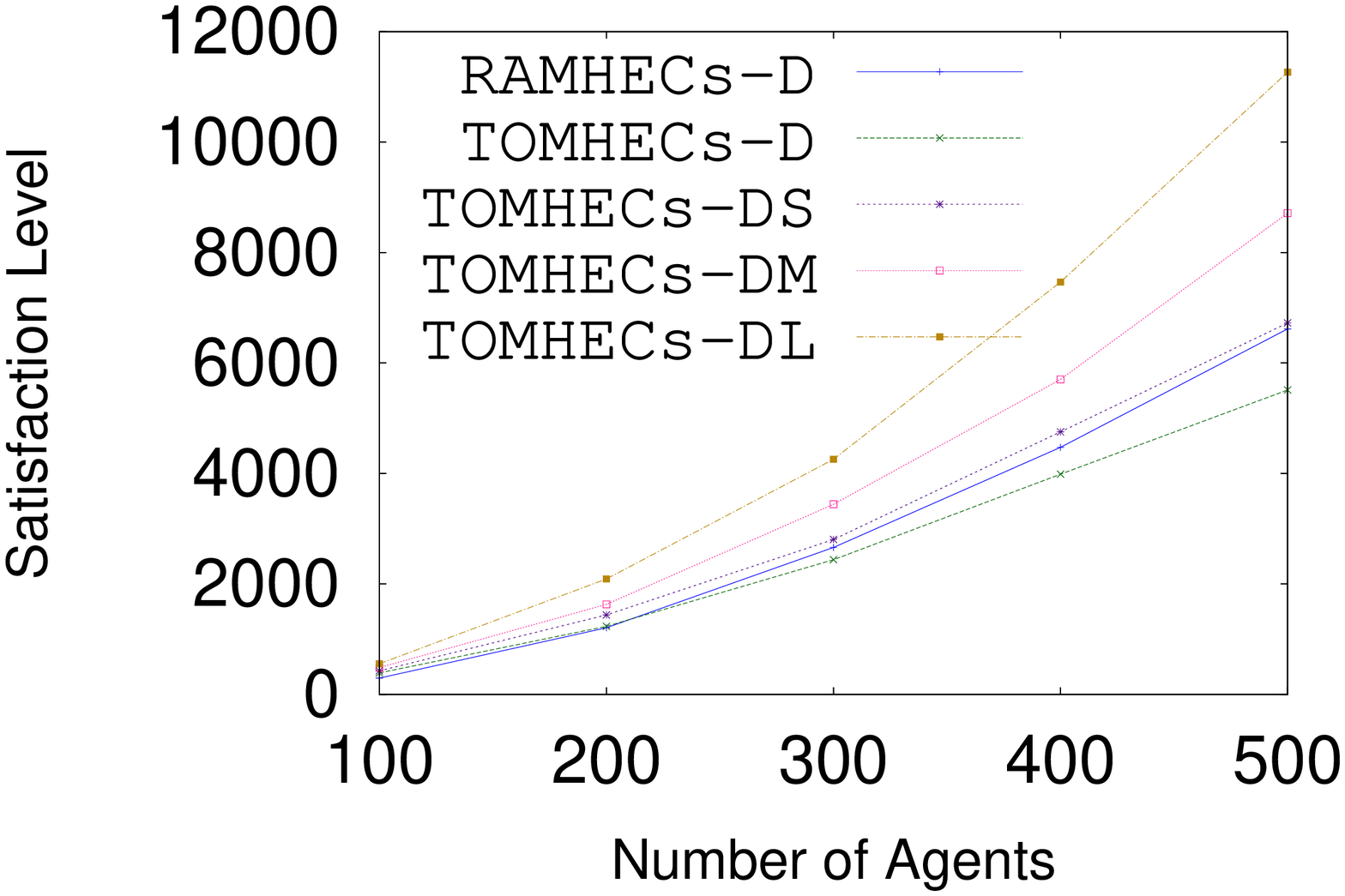}
                \subcaption{$\boldsymbol{\eta_{\ell}}$ of doctors}
                \label{fig:sim4a}
        \end{subfigure}
        \caption{$\boldsymbol{\eta_{\ell}}$ of requesting party with $m==n$}
\end{figure} 
\noindent Under the \emph{manipulative} environment of the \emph{requesting party}, it can be seen in Figure \ref{fig:sim3a} and Figure \ref{fig:sim4a} that, the \emph{satisfaction level} of the system in case of TOMHECs with large variation is less than the \emph{satisfaction level} 
\begin{figure}[H]
\begin{subfigure}[b]{0.48\textwidth}
                \centering
                \includegraphics[scale=0.17]{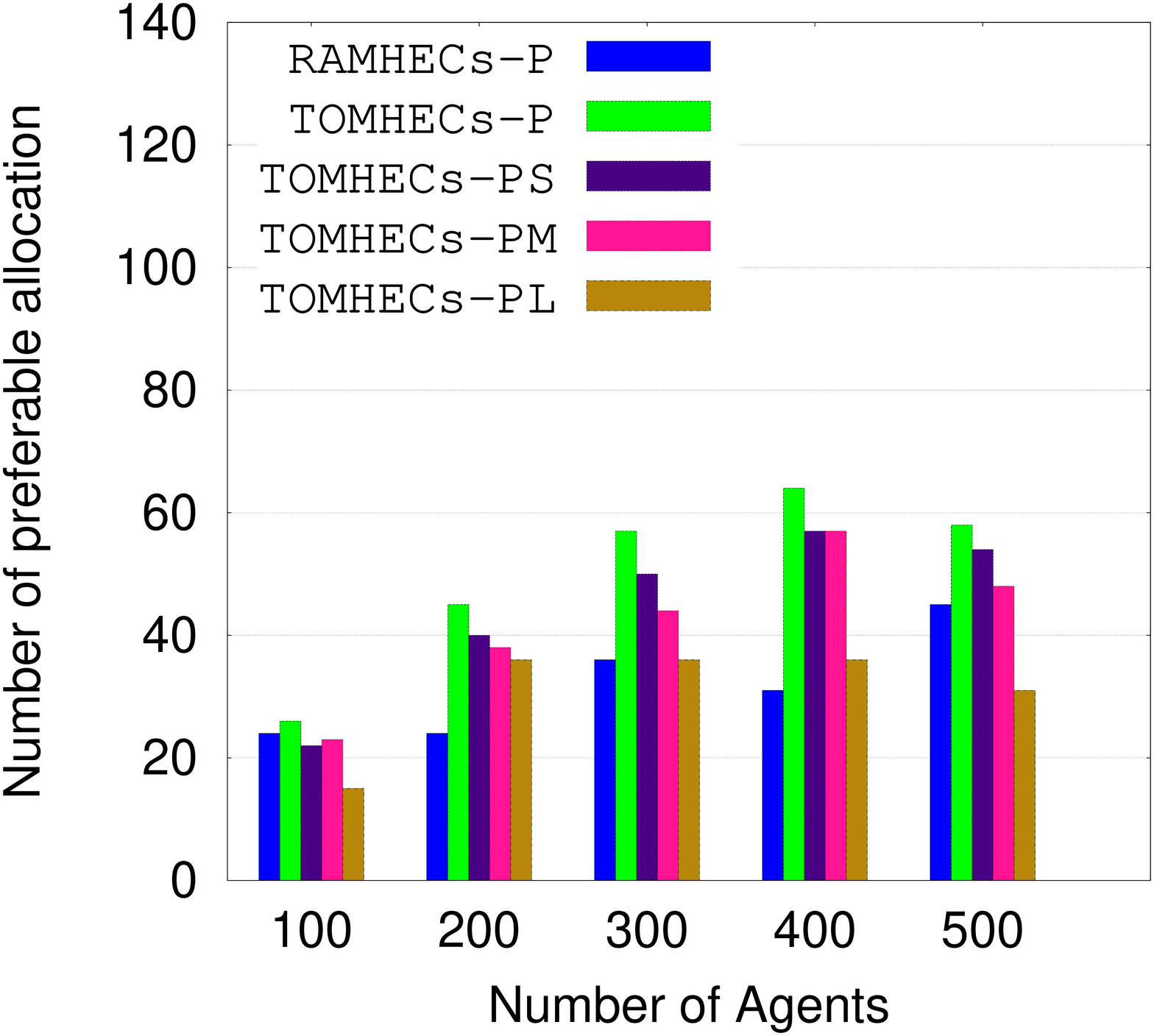}
                \subcaption{$\boldsymbol{\zeta}$ of patients}
                \label{fig:sim3b}
        \end{subfigure}%
        \begin{subfigure}[b]{0.48\textwidth}
                \centering
                \includegraphics[scale=0.17]{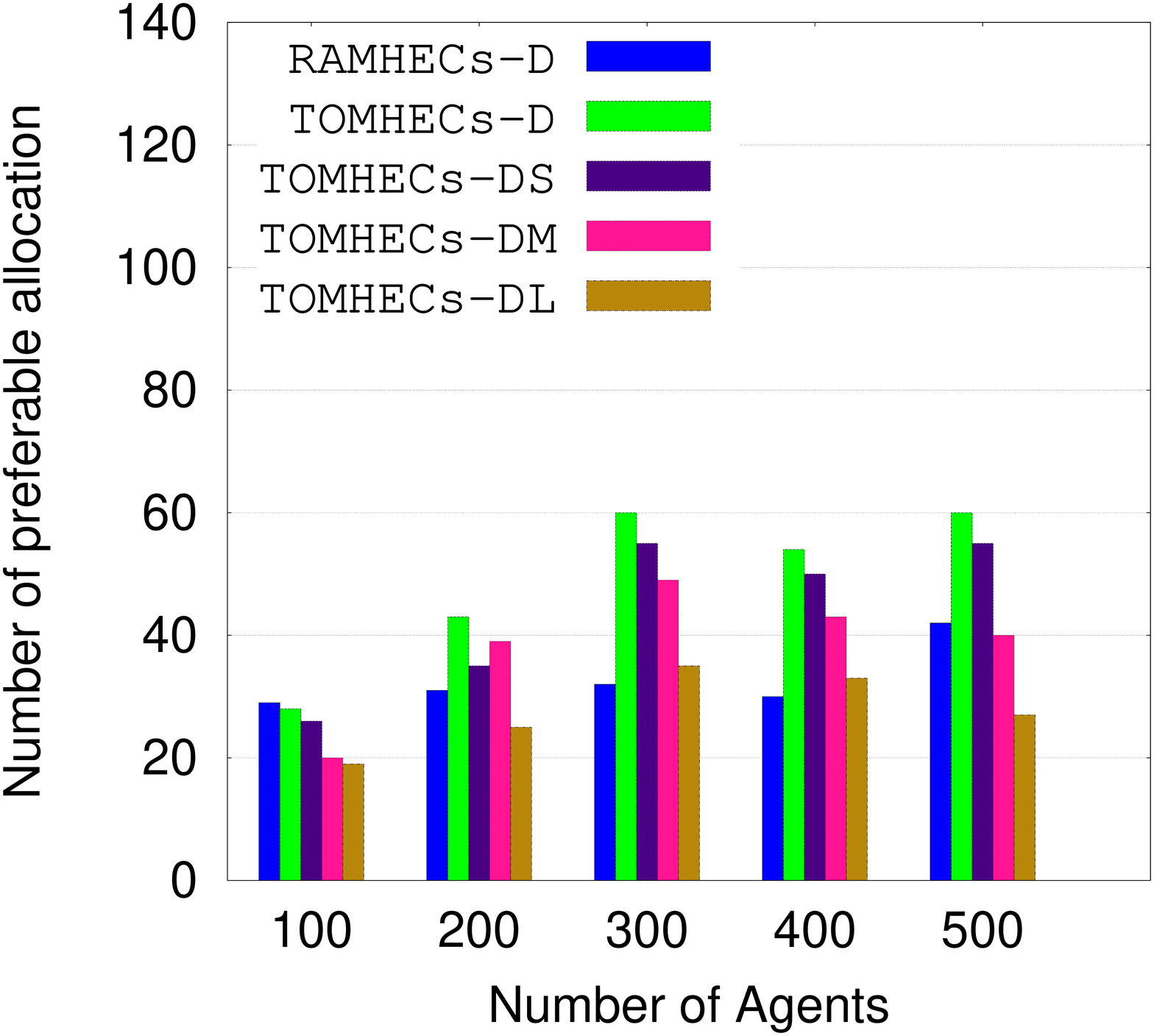}
                \subcaption{$\boldsymbol{\zeta}$ of doctors}
                \label{fig:sim4b}
        \end{subfigure}
        \caption{$\boldsymbol{\zeta}$ of requesting party with $m==n$}
\end{figure}

\noindent of the system in case of TOMHECs with medium variation and even less than RAMHECs is less than the \emph{satisfaction level} of the system in case of TOMHECs with small variation even less than RAMHECs is less than the \emph{satisfaction level} of the system in case of TOMHECs. It is natural from the construction of TOMHECs. Considering the second parameter $i.e.$ \emph{number of preferable allocation}, it can be seen in Figure \ref{fig:sim3b} and Figure \ref{fig:sim4b} that the NPA of the requesting party in case of TOMHEcs is more as compared to RAMHECs. Under the \emph{manipulative} environment of the \emph{requesting party}, it can be seen in Figure \ref{fig:sim3b} and Figure \ref{fig:sim4b} that, the NPA of the system in case of TOMHECs with large variation is less than the NPA of the system in case of TOMHECs with medium variation is less than the NPA of the system in case of TOMHECs with small variation is less than the NPA of the system in case of TOMHECs.\\
 
$\bullet$ \textbf{Case 2a: Requested party with full preference (FP)}
In Figure \ref{fig:sim5a}, Figure \ref{fig:sim6a} and Figure \ref{fig:sim5b}, Figure \ref{fig:sim6b}, it can be seen that the \emph{satisfaction} \emph{level} and the NPA respectively of the requested party in case of TOMHECs is more as compared to RAMHECs. It can be seen from Figure \ref{fig:sim1a}-\ref{fig:sim2b} and Figure \ref{fig:sim5a}-\ref{fig:sim6b} that the TOMHECs is requesting party optimal.  It is natural from the construction of TOMHECs. 
\begin{figure}[H]
\begin{subfigure}[b]{0.48\textwidth}
                \centering
                \includegraphics[scale=0.30]{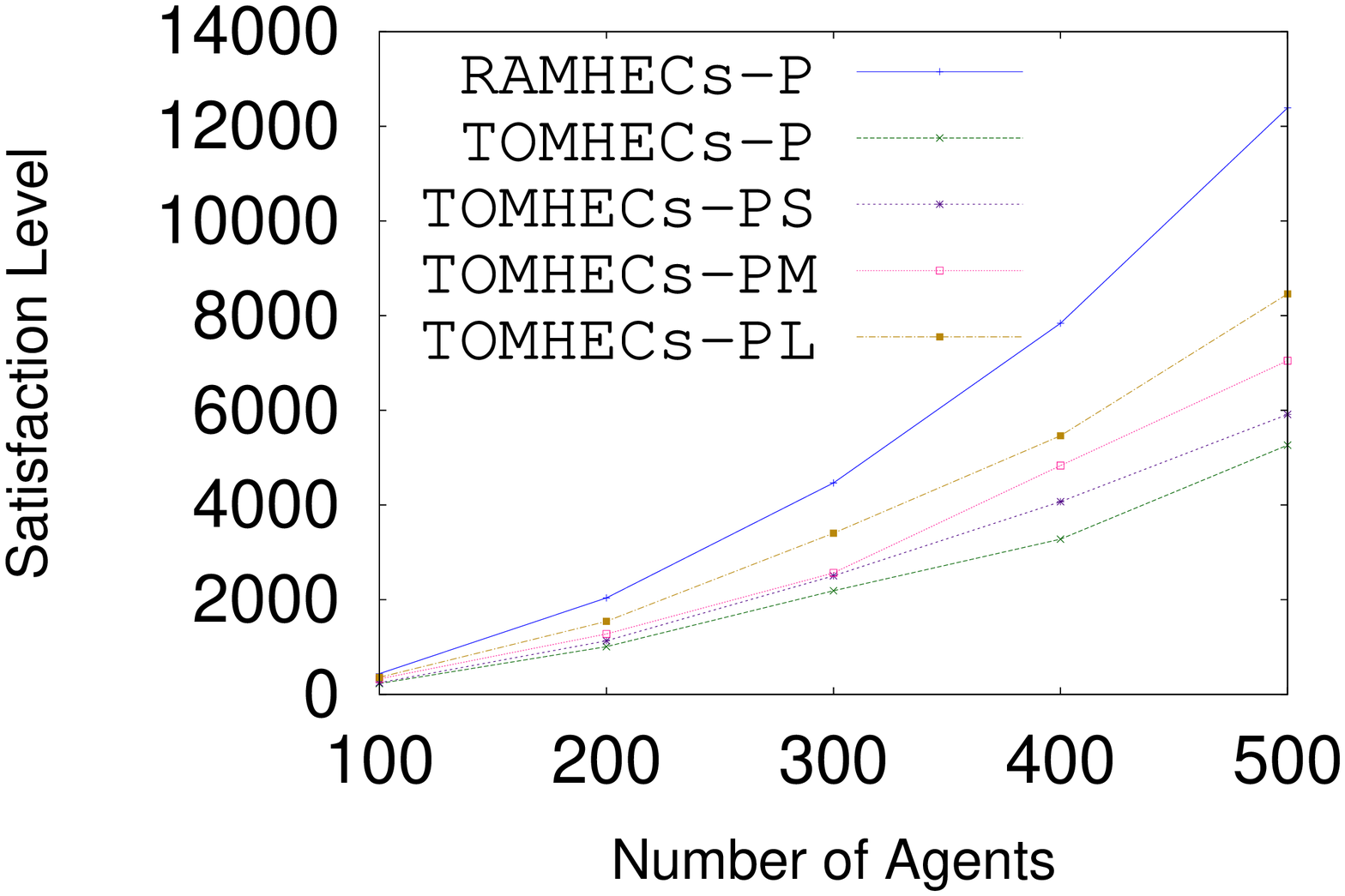}
                \subcaption{$\boldsymbol{\eta_{\ell}}$ of patients}
                \label{fig:sim5a}
        \end{subfigure}%
        \begin{subfigure}[b]{0.48\textwidth}
                \centering
                \includegraphics[scale=0.30]{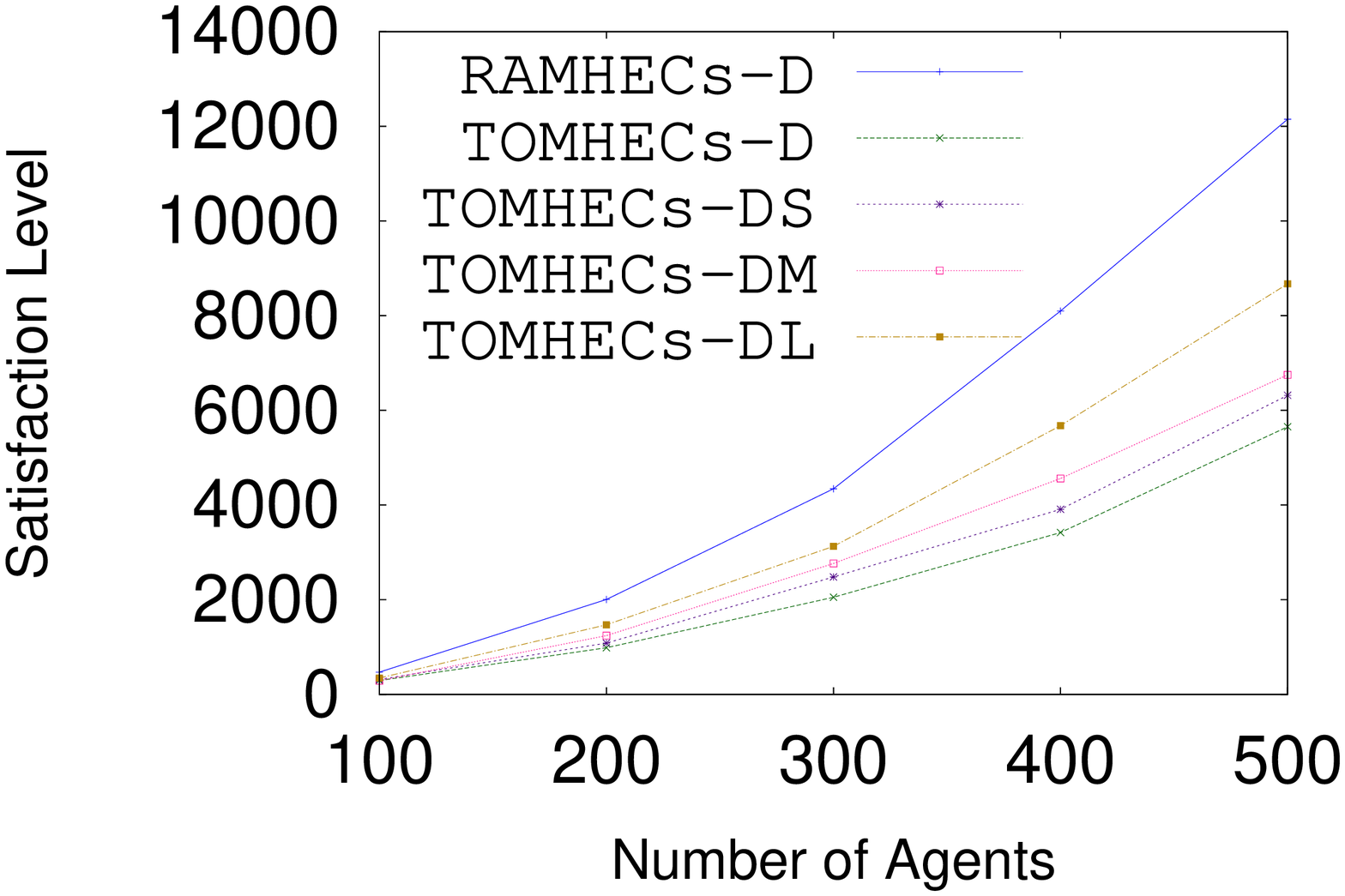}
                \subcaption{$\boldsymbol{\eta_{\ell}}$ of doctors}
                \label{fig:sim6a}
        \end{subfigure}
        \caption{$\boldsymbol{\eta_{\ell}}$ of requested party with $m==n$}
\end{figure} 

\begin{figure}[H]
\begin{subfigure}[b]{0.48\textwidth}
                \centering
                \includegraphics[scale=0.17]{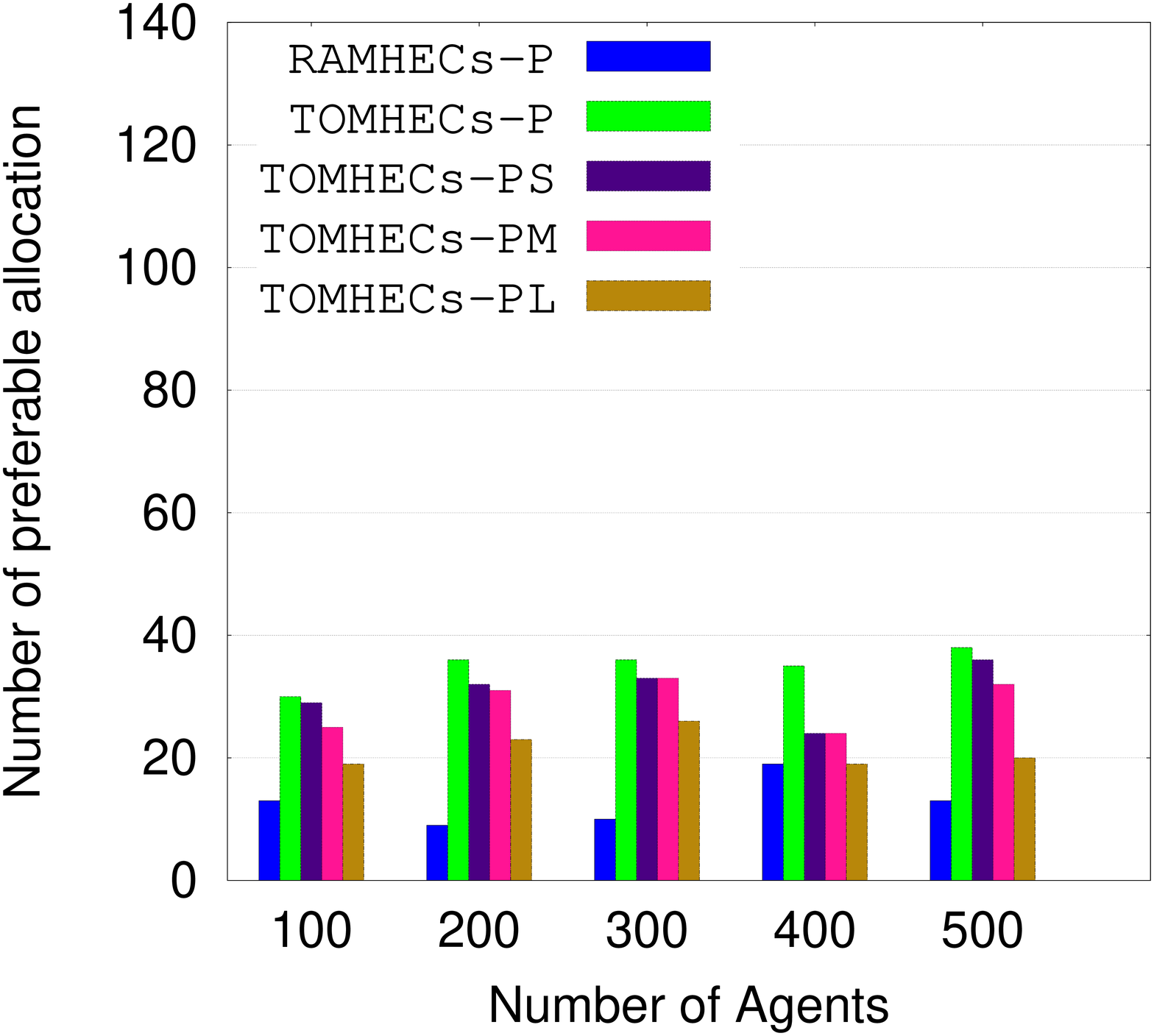}
                \subcaption{$\boldsymbol{\zeta}$ of patients}
                \label{fig:sim5b}
        \end{subfigure}%
        \begin{subfigure}[b]{0.48\textwidth}
                \centering
                \includegraphics[scale=0.17]{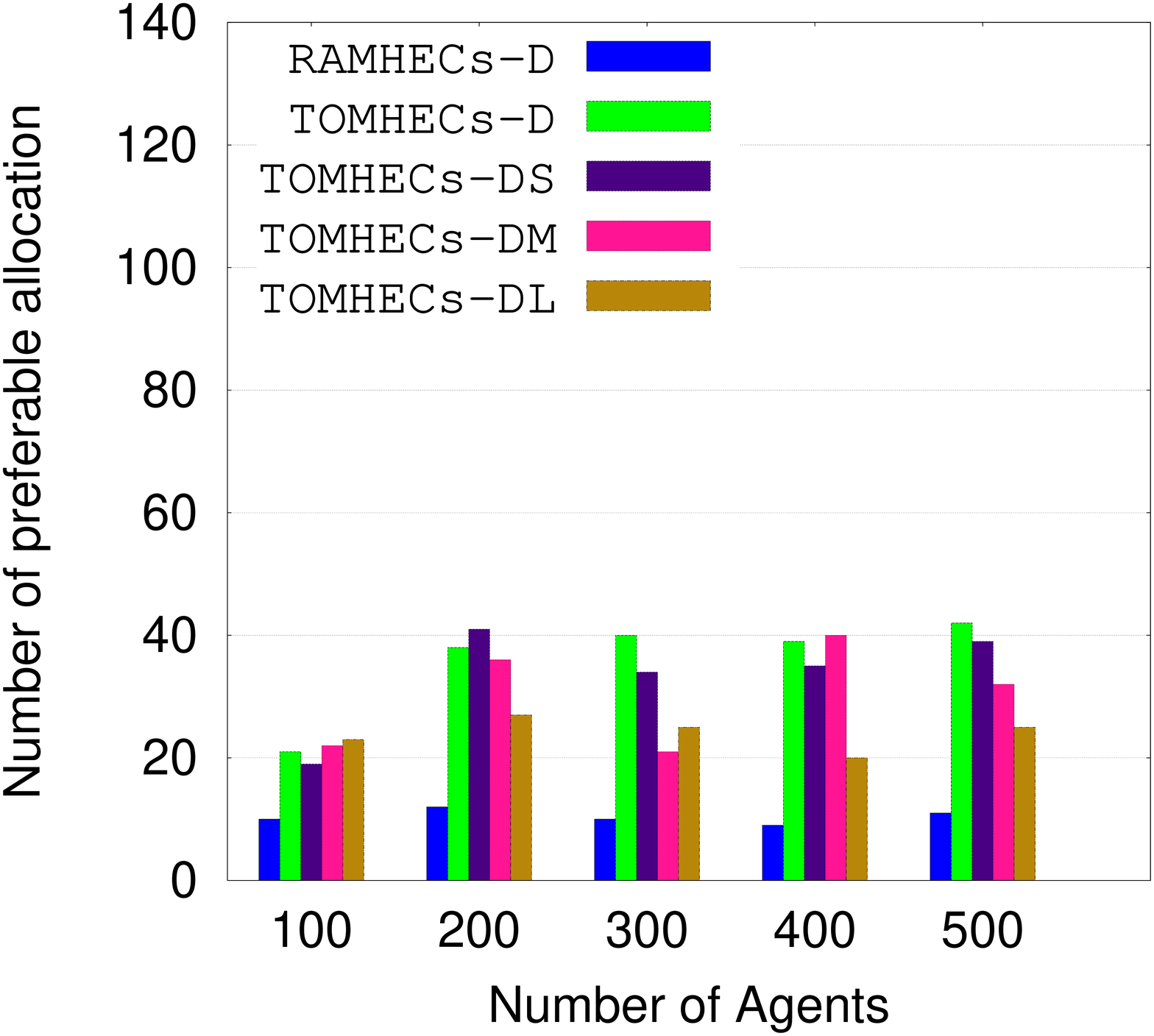}
                \subcaption{$\boldsymbol{\zeta}$ of doctors}
                \label{fig:sim6b}
        \end{subfigure}
        \caption{$\boldsymbol{\zeta}$ of requested party with $m==n$}
\end{figure}
 
$\bullet$ \textbf{Case 2b: Requested party with partial preference (PP)}
In Figure \ref{fig:sim7a}, Figure \ref{fig:sim8a} and Figure \ref{fig:sim7b}, Figure \ref{fig:sim8b}, it can be seen that the \emph{satisfaction level} and the NPA respectively of the requested party in case of TOMHECs is more as compared to RAMHECs. 

\begin{figure}[H]
\begin{subfigure}[b]{0.48\textwidth}
                \centering
                \includegraphics[scale=0.30]{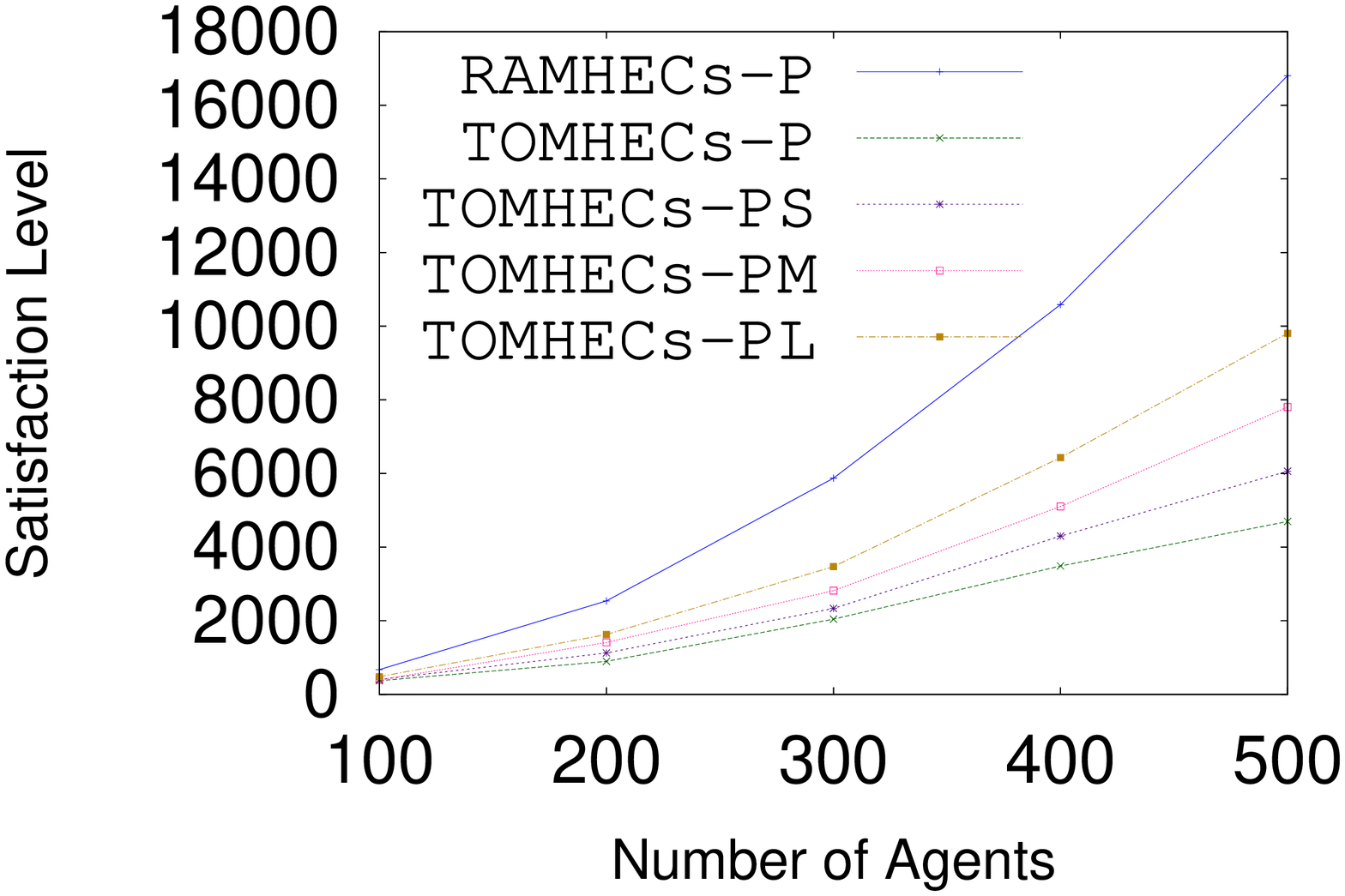}
                \subcaption{$\boldsymbol{\eta_{\ell}}$ of patients}
                \label{fig:sim7a}
        \end{subfigure}%
        \begin{subfigure}[b]{0.48\textwidth}
                \centering
                \includegraphics[scale=0.30]{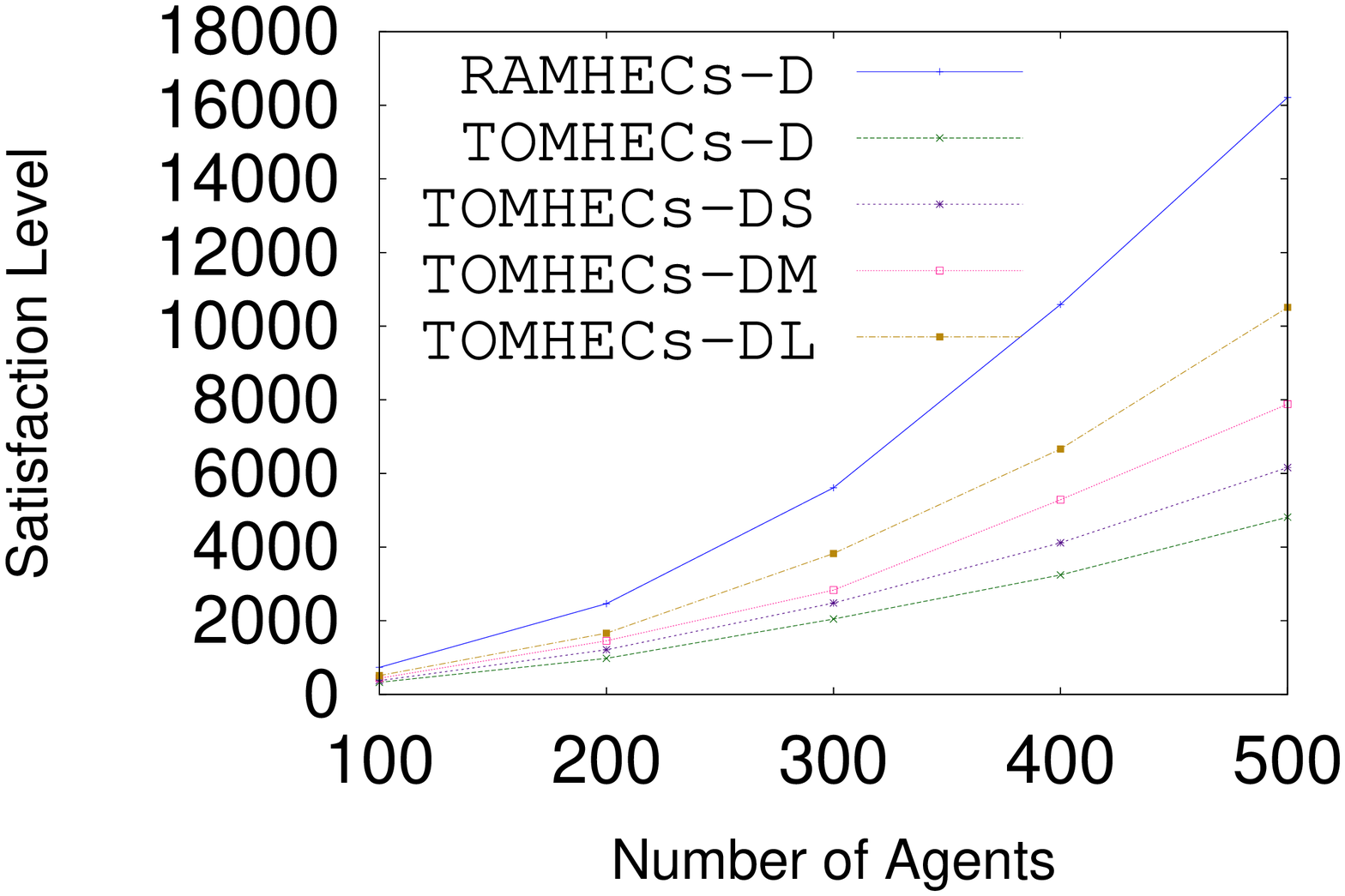}
                \subcaption{$\boldsymbol{\eta_{\ell}}$ of doctors}
                \label{fig:sim8a}
        \end{subfigure}
        \caption{$\boldsymbol{\eta_{\ell}}$ of requested party with $m==n$}
\end{figure} 

\begin{figure}[H]
\begin{subfigure}[b]{0.48\textwidth}
                \centering
                \includegraphics[scale=0.17]{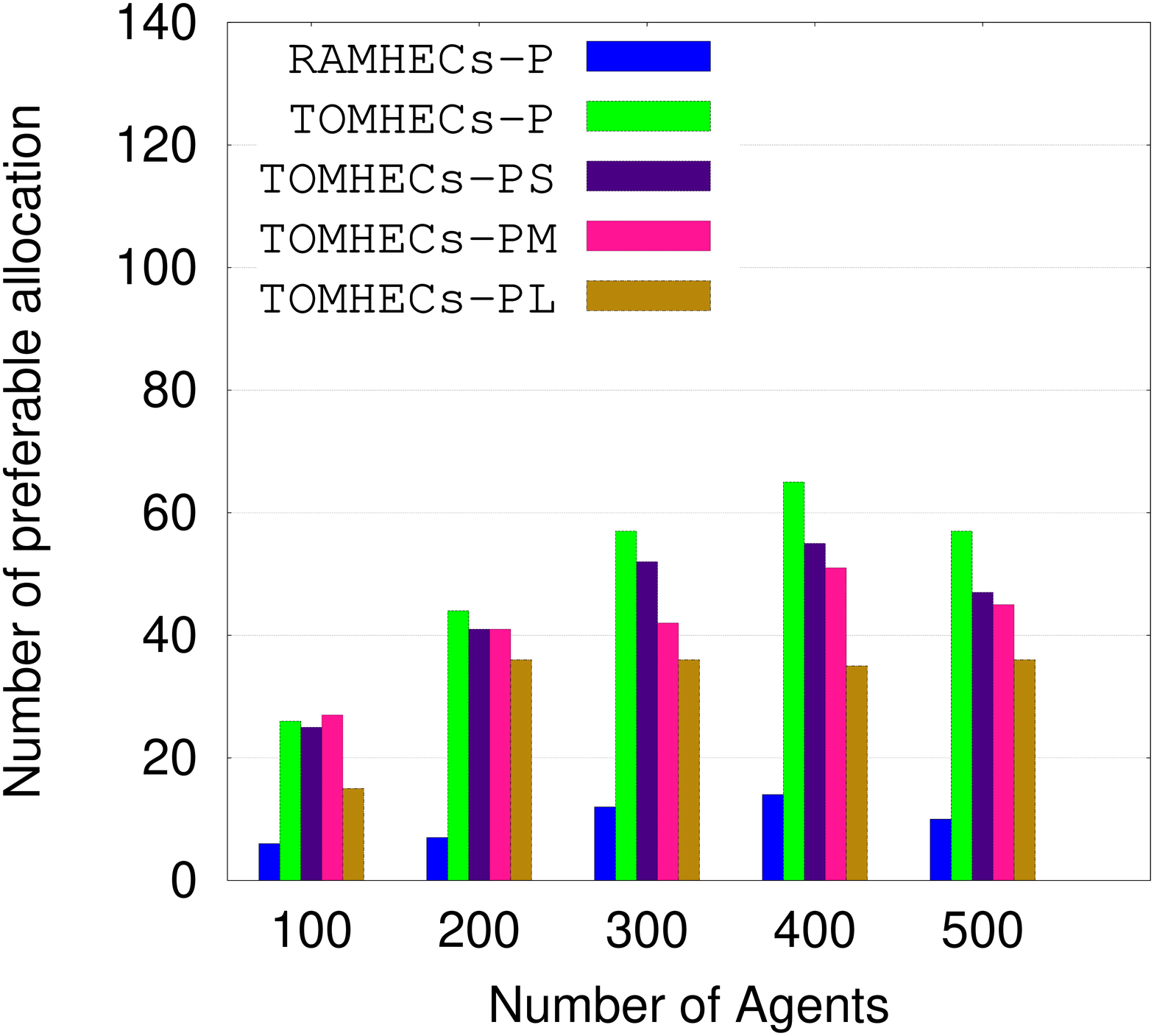}
                \subcaption{$\boldsymbol{\zeta}$ of patients}
                \label{fig:sim7b}
        \end{subfigure}%
        \begin{subfigure}[b]{0.48\textwidth}
                \centering
                \includegraphics[scale=0.17]{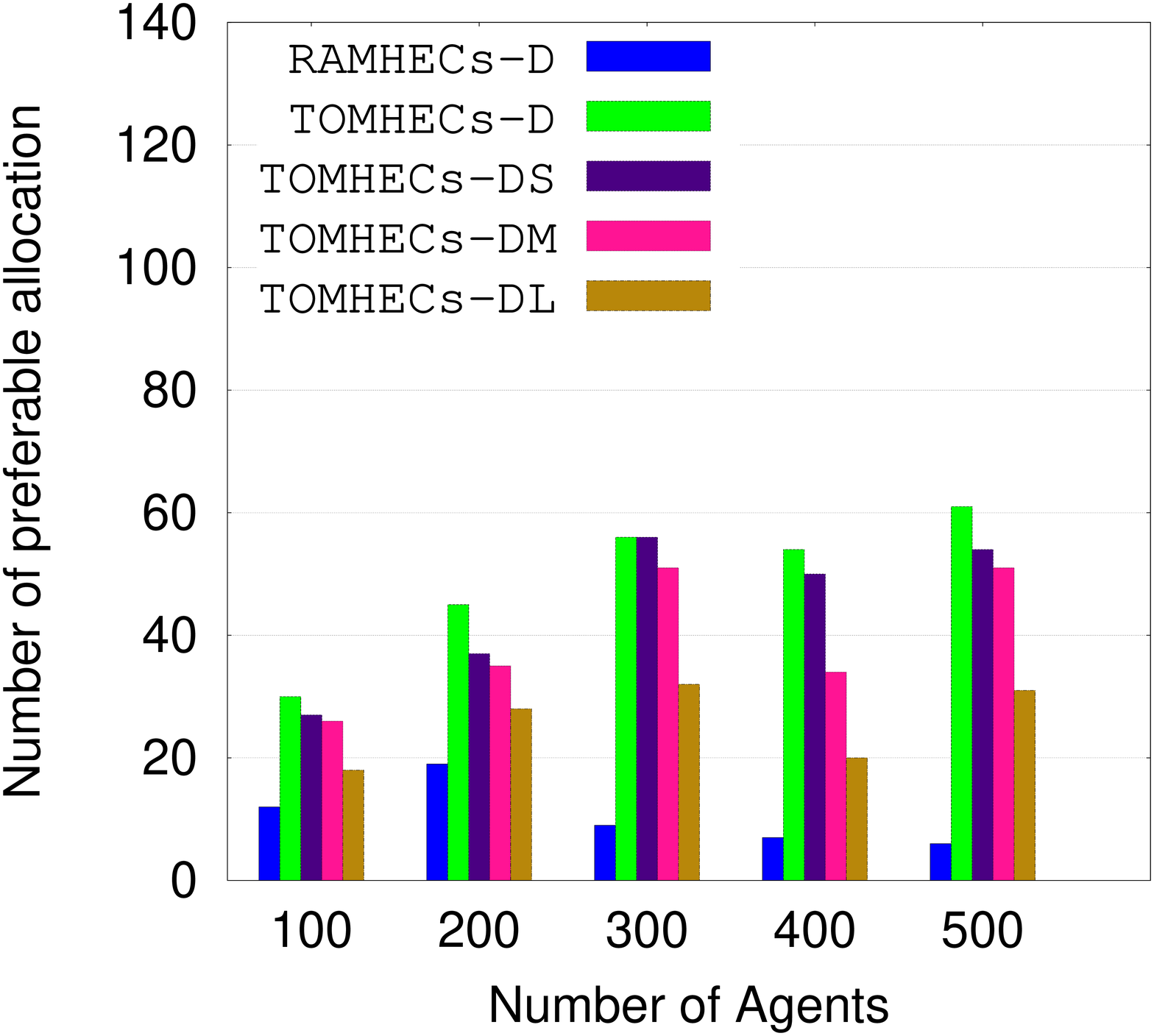}
                \subcaption{$\boldsymbol{\zeta}$ of doctors}
                \label{fig:sim8b}
        \end{subfigure}
        \caption{$\boldsymbol{\zeta}$ of requested party with $m==n$}
\end{figure}

\section{Conclusions and future works}
We have tried to model the \emph{ECs hiring problem} as a two sided matching problem in healthcare domain. This paper proposed an \emph{optimal} and \emph{truthful} mechanism, namely TOMHECs to allocate the ECs to the patients.
The more general settings are of \emph{n} patients and \emph{m} doctors ($m \neq n$ or $m == n$) with the constraint that members of the \emph{patient party} and \emph{doctor party} can provide the preference ordering (not necessarily strict) over the subset of the members of the opposite party can be thought of as our future work.
 \section*{Acknowledgement}
\noindent We would like to thank Prof. Y. Narahari and members of the Game Theory Lab. at Department of CSA, IISc Bangalore for their useful advices. We would like to thank the faculty members, and PhD research scholars of the department for their valuable suggestions. We highly acknowledge the effort undertaken by Ministry of Electronics and Information Technology, Media Lab Asia, Government of India through the Visvesvaraya scheme.


  \bibliographystyle{elsarticle-num} 
 \bibliography{phd}





\end{document}